\documentclass[twoside]{article}
\usepackage[letterpaper,margin=1in]{geometry}
\usepackage{enumitem}
\usepackage{amsthm}
\usepackage{tikz}
\usetikzlibrary{arrows.meta}
\usetikzlibrary{decorations.pathreplacing,calligraphy}
\usepackage[charter]{mathdesign}
\usepackage[linktocpage=true,pagebackref=true,colorlinks,linkcolor=magenta,citecolor=blue,bookmarks,bookmarksopen,bookmarksnumbered]{hyperref}
\usepackage[boxruled,vlined,linesnumbered]{algorithm2e}
\usepackage{lineno}
\SetKwRepeat{Do}{do}{while}
\usepackage{comment}
\usepackage{eulervm}
\usepackage{booktabs} 
\usepackage{multirow}
\usepackage{quoting}
\usepackage{placeins}
\usepackage{subcaption}
\usepackage{colortbl}
\usepackage{xcolor}
\usepackage{adjustbox}
\usepackage{float}
\usepackage{graphicx}
\usepackage{capt-of}
\usepackage{wrapfig}
\usepackage{tabularx}
\usepackage{caption}
\usepackage{amsmath}
\usepackage[nameinlink]{cleveref}
\usepackage{makecell} 
\renewcommand{\arraystretch}{1.3}
\usepackage{tcolorbox}
\usepackage[framemethod=TikZ]{mdframed}
\usetikzlibrary{arrows.meta, calc, shapes.geometric}

\newtheorem{theorem}{Theorem}[section]
\newtheorem{lemma}[theorem]{Lemma}

\newtheorem{definition}[theorem]{Definition}

\newcommand{\dia}{\diamond}
\newcommand{\al}{{\alpha}}
\newcommand{\be}{{\beta}}

\newcommand{\todo}[1]{{\color{red} [#1]}}
\newcommand{\ft}{\textsc{Ft}}
\newcommand{\ep}{{\varepsilon}}
\newcommand{\Ep}{{1+\varepsilon}}
\newcommand{\loge}{\log_\Ep nW}
\newcommand{\logee}{\log_\Ep^2 nW}
\newcommand{\logt}{\log_\Ep^3 nW}
\newcommand{\logf}{\log_\Ep^4 nW}
\newcommand{\fpa}{{\textsc{Path}_1}}
\newcommand{\spa}{{\textsc{Path}_2}}
\newcommand{\tpa}{{\textsc{Path}_3}}

\newcommand{\pa}{\textsc{{Path}}}
\newcommand{\bu}{\textsc{DetSegment}}
\newcommand{\quN}{\textsc{Query}_0}
\newcommand{\quDone}{\textsc{QueryDetour}}
\newcommand{\quDtwo}{\textsc{QueryDual}}
\newcommand{\mn}{\textsc{MakeNode}}
\newcommand{\ro}{\textsc{Root}}
\newcommand{\pre}{\textsc{Prefix}}
\newcommand{\suf}{\textsc{Suffix}}
\newcommand{\rr}{{\mathcal{R}}}
\newcommand{\rri}{{\rr_i}}
\newcommand{\xx}{{\mathcal{X}}}
\newcommand{\yy}{{\mathcal{Y}}}
\newcommand{\pr}{\mathbf{P}}
\newcommand{\tl}{{\tilde{O}}}
\newcommand{\stf}{{st \dia F}}
\newcommand{\qu}{\textsc{Query}}
\newcommand{\std}{{st \dia F}}
\SetKwComment{Comment}{\hfill\textcolor{blue}{$\triangleright$~}}{}
\newcommand{\ado}{\textsc{Do}}
\newcommand{\sdo}{\textsc{Sdo}}
\def\myone{1}
\title{Approximate Single Source Dual Fault Tolerant Distance Oracle%
\thanks{Supported by ANRF grant ANRF/ARGM/2025/001853/MTR.}}
    \author{
  \ Koustav Das \\
  IIT Gandhinagar \\
  India\\
  \texttt{koustav.das@iitgn.ac.in} \\
  \and
  \ Manoj Gupta \\
  IIT Gandhinagar \\
  India\\
  \texttt{gmanoj@iitgn.ac.in} \\
} 

\begin{document}
\setlength{\abovedisplayskip}{1pt}
\setlength{\belowdisplayskip}{5pt}

\date{}

\maketitle
\begin{abstract}
We are given an undirected weighted graph $G$ with $n$ vertices and $m$ edges, edge weights in $[1, W]$, and a designated source vertex $s$. We design a single source dual fault tolerant distance oracle for $G$. Given a destination vertex $t$ and a set $F$ of at most two faulty edges, the oracle returns a $(1 + O(\ep))$-approximation of the weight of the shortest path from the source $s$ to $t$ avoiding $F$. Our oracle uses $\tl(n\sqrt{n})$ space\footnote{$\text{poly}$$(\log_{\Ep} nW)$ factors are hidden in the $\tl$ notation.} and has $\tl(1)$ query time.

Prior to our result, single source {\em single} fault tolerant oracles were known to return a $(1+\ep)$ approximation of the weight of the shortest path using $\tl(n)$ space and $O(1)$ query time. However, extending these approaches to multiple faults remained an open problem. Indeed, all  $(1+\ep)$-approximate  distance oracles that handle multiple faults require $\Omega(n^2)$ space. We break this bound by presenting the first dual fault tolerant distance oracle with $o(n^2)$ space.

\end{abstract}

\section{Introduction}

In many real-world scenarios, we are interested in computing the shortest distances from a specific location to all other reachable locations in a network. Such networks can be modelled as graphs, where vertices represent places and edges capture connectivity between them. This naturally leads to the classical \textit{single source shortest paths} (SSSP) problem, where the objective is to preprocess a graph so that distances from a designated source $s$ to any other vertex can be efficiently retrieved in response to queries.

\[
\qu(s,t): \text{ Find the weight of the shortest path from } s \text{ to } t.
\]

For unweighted graphs, the single source shortest paths (SSSP) can be efficiently computed using a breadth-first search (BFS). The distances from the source to all other vertices can be stored in $O(n)$ space, allowing $O(1)$ query time per vertex, where $n$ and $m$ denote the number of vertices and edges, respectively. For weighted graphs, classical SSSP algorithms such as Dijkstra's algorithm allow us to preprocess the graph in $O(m+n\log n)$ time so that distances from a designated source can be queried in constant time. Over the years, a rich body of work has advanced the classical SSSP problem~\cite{fredman1993surpassing,fredman1990trans,thorup2000ram,raman1996priority,raman1997recent,hagerup2000improved}. 
Later, Duan, Mao, Shu, and Yin~\cite{duan2023randomized} developed a randomized algorithm for the SSSP problem with a time complexity of $ O(m \sqrt{\log n \log \log n}) $.
Recently, Duan, Mao, Mao, Shu, Yin~\cite{duan2025breaking} achieved further progress by presenting a deterministic $O(m \log^{2/3} n)$-time algorithm for directed graphs with non-negative edge weights.

However, in practice, some connections may fail or become temporarily unavailable. We model these disruptions as \emph{faulty edges}, i.e., edges removed from the graph. The task is to determine the length of the  shortest path from a specified source vertex $s$ to any other vertex $ t \in V  $
while avoiding faulty edges. Our objective is to answer the following query efficiently:

\begin{center}
	$\qu(t, F)$: Find the weight of the shortest path from $s$ to $t$ avoiding all the edges contained in $F$
\end{center}

An algorithm may preprocess the graph $G$ and create a data structure to quickly answer the above query. Such an algorithm (and the data structure) is called a {\em distance oracle} in literature. Since we are dealing with a single source $s$ and must output shortest paths avoiding faults, our oracle is called {\em Single Source fault tolerant} distance oracle.

Let us first see a simple single source fault tolerant distance oracle.   One could naively precompute and store exact shortest paths from the source $s$ to all other vertices under all possible faulty edge scenarios. However, this approach is computationally prohibitive, as both the preprocessing time and space taken by the data structure grow rapidly with the number of faults. For example, even with only two faulty edges, the space requirement already becomes $O(m^2 n)$. Another naive approach is not to preprocess the graph at all.  Upon  $\qu(t,F)$, we run Dijkstra's algorithm from $s$ after removing the faulty edges from the graph. This distance oracle is space-efficient, but the query time is too high.

In our paper, we focus primarily on the single source fault tolerant distance oracle. However, in the literature, a lot of work has also been done on the all-pairs variant, where any vertex can serve as the source. To distinguish between these variants, we use the following notation:

\begin{itemize}[noitemsep]
    \item $\sdo(f)$: A single source fault tolerant distance oracle handling up to $f$ faults.
    \item $\ado(f)$: An all-pairs fault tolerant distance oracle handling up to $f$ faults.
\end{itemize}

Distance oracles are further classified by the quality of their output. We say an $\sdo(f)$ is \textbf{$\al$-approximate} if, for all $t \in V$ and all valid fault sets $F$ ($|F| \le f$), the query returns a value that is at least the length of the shortest path from $s$ to $t$ avoiding $F$, and at most $\al$ times this length. If $\al=1$, the oracle\footnote{For exact distance oracle, {$\tl$ notation hides poly $\log n$ factor}} is \textbf{exact}. These definitions naturally extend to $\ado(f)$.

We now summarize key results on single source and (general) distance oracles under failures in undirected graphs.
Gupta and Singh~\cite{GuptaS18} designed an $\sdo(1)$ of size $\tilde{O}(n^{3/2})$ that answers queries in $\tilde{O}(1)$ time for an undirected unweighted graph. Later, Bil\`{o}, Cohen, Friedrich, and Schirneck~\cite{BiloCFS21} extended this result to weighted graphs with an oracle size of $O(W^{1/2}n^{3/2})$, where $W$ is the maximum edge weight. Next, we discuss results on approximate distance oracles.

Baswana and Khanna~\cite{BaswanaK13} initiated the systematic study of approximate fault tolerant distance oracles. 
They developed an $O(m \log n)$-time constructible $\sdo(1)$ of size $O(n \log n)$ reporting a $3$-approximate shortest-path distance from a fixed source to any vertex $v \in V$ while avoiding any failed vertex $x \in V$. 
For undirected, unweighted graphs, they also proposed an 
$O\left(\frac{n \log n}{\epsilon^{3}}\right)$-space $\sdo$(1) that reports $(1+\epsilon)$-approximate shortest paths for any $\epsilon > 0$. More recently, ~\cite{baswana2020approximate} constructed an $\sdo(1)$ of size $O\left(n \log_{1+\epsilon} (nW)\right)$ with query time 
$O\left(\log \log_{1+\epsilon}(nW)\right)$ for directed weighted graphs. 
	\ifnum\myone=0
		We provide a summary of results in the full version of the paper below in Table 1.
	\else
		We provide a summary of results in \Cref{tab:ResultsFaultTolerantDistanceOracle}.
	\fi

\ifnum\myone=1
\begin{table}[ht]
\centering
\renewcommand{\arraystretch}{1.2}

\begin{tabular}{p{1cm} p{3cm} p{3cm} p{1.8cm} p{2cm} p{2cm}}
\toprule
\textbf{Oracle} & \textbf{Space} & \textbf{Query} & \textbf{Approx.} & \textbf{Remarks} & \textbf{Ref.} \\
\midrule

$\sdo(1)$ & $\tilde{O}(n^{3/2})$ & $\tilde{O}(1)$ & $1$ & Undirected, Unweighted & \cite{GuptaS18} \\

$\sdo(1)$ & $O(W^{1/2}n^{3/2})$ & $\tilde{O}(1)$ & $1$ & Undirected, Weighted & \cite{BiloCFS21} \\

$\sdo(1)$ & $O(n\log n)$ & $O(1)$ & $3$ & Undirected, Weighted for vertex failure & \cite{BaswanaK13} \\

$\sdo(1)$ & $O\!\left(\frac{n \log n}{\epsilon^{3}}\right)$ & $O(1)$ & $(1+\epsilon)$ & Undirected, Unweighted for vertex failure & \cite{BaswanaK13} \\

$\sdo(1)$ & $O\!\left(n \log_{1+\epsilon}(nW)\right)$ & $O\!\left(\log \log_{1+\epsilon}(nW)\right)$ & $(1+\epsilon)$ & Directed, Weighted for vertex or edge failure & \cite{baswana2020approximate} \\

$\sdo(2)$ & $O(n\sqrt{n}\log_{1+\epsilon}^4 nW)$ & $O(\log_{1+\epsilon}^{c} nW)$ & $(1+O(\epsilon))$ & Undirected, Weighted & \textcolor{blue}{Our Result} \\

\bottomrule
\end{tabular}

\caption{Summary of fault tolerant distance oracle results.}
\label{tab:ResultsFaultTolerantDistanceOracle}
\end{table}
\fi
One may ask if there are any result on exact or $(\Ep)$-approximate $\sdo(f)$ for $f\ge 2$. Yes, there are, but only indirectly. To this end, notice that  a $\ado(f)$ is also an $\sdo(f)$. We now show some relevant results related to $\ado(f)$ where $f \ge 2$.  Duan and Pettie~\cite{DuanP09} designed an exact $\ado(2)$ of size $\tl(n^2)$ with $\tl(1)$ query time. Chechik, Cohen, Fiat, Kaplan~\cite{ChechikCFK17} designed a $(1+\epsilon)$-approximate $\ado(f)$ that tolerates up to $f = O\left(\frac{\log n}{\log \log n}\right)$ faults with $\tl(n^2)$ space and $O(f^5)$ query time. Both these algorithms also imply the existence of exact and $(\Ep)$-approximate $\sdo(f)$, though requiring $\tl(n^2)$ space. One of the main questions in this field is to design an $\sdo(2)$ (and in general $\sdo(f)$ for $f \ge 2$) with sub-quadratic space.

 In this work, we narrow this gap by studying the single source setting under up to two edge failures. To the best of our knowledge, we provide the first approximate $\sdo(2)$ that requires only $o(n^2)$ space while supporting queries in nearly constant time.

\subsection{Our Result}

\begin{theorem}\label{thm:ft-sssp}

Given an undirected graph $G$ with  edge weights in $[1,W]$, there is a $(1+O(\ep))$ approximate $\sdo(2)$ of size $O(n\sqrt{n} \log_{1+\ep}^4 nW)$ and $O(\log_{1+\ep}^{c} nW)$ query time where $c>0$ is a constant.

\end{theorem}

\ifnum\myone=1
\subsection{Related Work}

Demetrescu, Thorup, Chowdhury, Ramachandran ~\cite{DemetrescuTCR08} designed a $\ado(1)$ of size $O(n^2 \log n)$ that can return exact distances in constant query time, i.e., $O(1)$. However, their approach incurs a substantial preprocessing cost of $O(mn^2 + n^3 \log n)$. To address this limitation, Bernstein and Karger~\cite{BernsteinK08, BernsteinK09} proposed improved constructions that reduced the preprocessing time to $O(n^2 \sqrt{m})$ and further to $\tilde{O}(mn)$, while keeping both the space complexity and query time unchanged. A major advancement was made by Grandoni and Williams~\cite{GrandoniW19}, who introduced the first $\ado(1)$ with subcubic preprocessing time and sublinear query time, marking a significant step forward in efficiency. Building on this, Chechik and Cohen~\cite{ChechikC20} retained subcubic preprocessing while improving the query time to $\tilde{O}(1)$, effectively combining fast preprocessing with near-constant query performance. More recently, Gu and Ren \cite{GuR21} further enhanced the trade-off by developing a $\ado(1)$ with a preprocessing time of $O(n^{2.6146}W)$ and constant query time, where $W$ denotes the maximum edge weight in the graph. The problem of multiple edge or vertex faults has also attracted significant attention. Duan and Pettie~\cite{DuanP09} studied the setting with two simultaneous failures, either vertices or edges, and proposed an $\ado(2)$ of size $O(n^2 \log^3 n)$ that can be constructed in polynomial time. Duan and Ren~\cite{DuanR22} developed an $\ado(f)$ for undirected weighted graphs that requires $O(fn^{4})$ space and has a query time of $f^{O(f)}$, where $f$ denotes the number of faulty edges.

We now look at some results for approximate $\ado(f)$. Chechik, Cohen, Fiat, Kaplan~\cite{ChechikCFK17} designed a $(1+\epsilon)$-approximate $\ado(f)$ that tolerates up to $f = O\left(\frac{\log n}{\log \log n}\right)$ faults, though at the cost of nearly quadratic space. Later, Bilò, Chechik, Choudhary, Cohen, Friedrich, Krogmann and Schirneck ~\cite{DBLP:journals/theoretics/BiloCCC0KS24} improved the space complexity by constructing a $(3+\epsilon)$-approximate oracle with subquadratic space. 

\fi

\section{Preliminaries}

Let $G = (V, E)$ be an undirected weighted graph with $n$ vertices and $m$ edges with edge weights in $[1\dots W]$. For any $s,t \in V$, $st$ denote a shortest path from $s$ to $t$, and $|st|$ its weight. For any path $P$, $|P|$ denotes the weight of the path $P$. Let $P[a,b]$ denotes the subpath from $a$ to $b$ in $P$. Given two paths $P_1$ and $P_2$ such that $P_1$ ends where $P_2$ begins, their concatenation is denoted $P_1 \odot P_2$. Let $st \dia F$ denotes the shortest path from $s$ to $t$ that avoids all edges in $F$, and $st \dia e$ when $F = \{e\}$ for any edge $e=(u,v)$.  All edges in $F$ are called \emph{faulty} or \emph{failed} edges. All vertices incident to edges in $F$, that is, the vertices in the set $V(F)$  are called \emph{affected} vertices. Given a path $P[s,t]$, we normally visualize it as drawn on a plane with $s$ at the top and $t$ at the bottom. For an edge $e = (u,v) \in P$ (where $u$ is closer to $s$), the term \emph{above $e$ on path $P$} refers to the subpath $P[s,u]$, and \emph{below $e$ on path $P$} refers to the subpath $P[v,t]$. Similarly, we can define the notion of {\em above/below of a vertex $v$} on a path $P$. There can be many faulty edges on the path $P$. The faulty edge closest to $t$ on the path $P$ is called the {\em last faulty} edge on $P$.

We assume that the shortest path between any pair of vertices in the graph is unique, even in the presence of edge failures. This can be ensured by applying a standard weight-perturbation technique in which each edge is assigned an infinitesimally small distinct weight (e.g., see~\cite{ParterP13}). As a result, any two distinct $st$ paths will have different total weights. Similar assumptions have been used in prior works such as~\cite{BernsteinK09,Hershberger2001,GuptaS18,DuanR22}.

\begin{definition}[$k$-decomposable paths]
    A path is \emph{$k$-decomposable} if it is the concatenation of at most $k+1$ shortest paths interleaved with at most $k$ edges, where $k\ge 0$. For example, $P_1 \odot e \odot P_2$ is a 1-decomposable path if both $P_1$ and $P_2$ are shortest paths in the graph $G$.
\end{definition}

   The following lemma formalizes the importance of decomposable paths.

   \begin{lemma}
       \label{lem:decompasable}(\cite{AfekBKCM02}) After $k$ edge failures in a weighted graph, each new shortest path is $k$-decomposable.
   \end{lemma}

 When $|F| = 2$, by the above lemma, the path $st \dia F$ can be written as the concatenation of at most three shortest paths interleaved with at most two edges. We call such a path \emph{2-decomposable}. Similarly, when $|F| = 1$, the path $st \dia F$ can be expressed as the concatenation of at most two shortest paths interleaved with at most one edge, referred to as a \emph{1-decomposable} path. We now redefine some notation from \cite{ChechikCFK17}.

\begin{definition}[Netpoints]
    Let $P$ be a (not necessarily shortest) path from $s$ to $t$ in $G$. We define the set $L = \{u_0, u_1, \dots\}$ such that $u_i$ is the first vertex on $P$ satisfying:
    \[
        |P[s, u_i]| \ge (\Ep)^i.
    \]
    If no such vertex exists for a given $i$, the sequence stops. Similarly, define $R = \{v_0, v_1, \dots\}$ as the set of vertices on $P$ where $v_i$ is the first vertex (traversing from $t$ to $s$) such that:
    \[
        |P[v_i, t]| \ge (\Ep)^i.
    \]
    We define the \emph{netpoints} of $P$ as the set $L \cup R \cup \{s,t\}$. The number of netpoints is $O(\log_{\Ep} |P|) = O(\log_{\Ep} nW)$.
\end{definition}

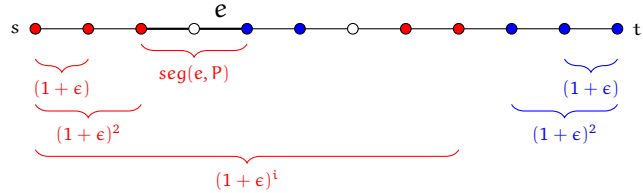
\begin{figure}[ht]
\centering
\begin{tikzpicture}[
    point/.style={circle,draw,fill=white,minimum size=4pt,inner sep=0pt},
    line/.style={thick},
    redpoint/.style={point,fill=red},
    bluepoint/.style={point,fill=blue},
    every label/.style={font=\scriptsize},
    >=Stealth
]


\draw[ thin] (0,0) -- (7.7,0);

\draw [decorate,decoration={brace,amplitude=5pt,mirror},red]
  (1.4,-0.2) -- (2.8,-0.2) node[midway,yshift=-0.4cm,red] {\scalebox{0.7}{$seg(e,P)$}};

\draw [decorate,decoration={brace,amplitude=5pt,mirror},red]
  (0,-0.4) -- (0.7,-0.4) node[midway,yshift=-0.4cm,red] {\scalebox{0.7}{$(1+\epsilon)$}};

\draw [decorate,decoration={brace,amplitude=5pt,mirror},red]
  (0,-1.0) -- (1.4,-1.0) node[midway,yshift=-0.4cm,red] {\scalebox{0.7}{$(1+\epsilon)^2$}};

\draw [decorate,decoration={brace,amplitude=5pt,mirror},red]
  (0,-1.6) -- (5.6,-1.6) node[midway,yshift=-0.4cm,red] {\scalebox{0.7}{$(1+\epsilon)^i$}};

\draw [decorate,decoration={brace,amplitude=5pt,mirror},blue]
  (7.0,-0.4) -- (7.7,-0.4) node[midway,yshift=-0.4cm,blue] {\scalebox{0.7}{$(1+\epsilon)$}};

\draw [decorate,decoration={brace,amplitude=5pt,mirror},blue]
  (6.3,-1.0) -- (7.7,-1.0) node[midway,yshift=-0.4cm,blue] {\scalebox{0.7}{$(1+\epsilon)^2$}};

\node[redpoint,label=left:$s$] (u) at (0,0) {};
\node[redpoint] (p1) at (0.7,0) {};
\node[redpoint] (p2) at (1.4,0) {};
\node[point] (p3) at (2.1,0) {};
\node[bluepoint] (p4) at (2.8,0) {};
\node[bluepoint] (p5) at (3.5,0) {};
\node[point,label=below:] (vj) at (4.2,0) {};
\node[redpoint,label=below:$$] (vj1) at (4.9,0) {};
\node[redpoint,label=below:$$] (p8) at (5.6,0) {};
\node[bluepoint] (p9) at (6.3,0) {};
\node[bluepoint] (p10) at (7.0,0) {};
\node[bluepoint,label=right:$t$] (v) at (7.7,0) {};
\draw[line] (p2) -- (p3);
\draw[line] (p3) -- (p4) node[midway, above=1pt] {$e$};

\end{tikzpicture}
\caption{Netpoints originating from the vertex \(s\) are shown in red and form the set \(L\). 
Netpoints originating from the vertex \(t\) are shown in blue and form the set \(R\). 
A vertex may belong to both \(L\) and \(R\). 
The union of these sets constitutes the set of netpoints.}

\label{fig:path-netpoints}
\end{figure}

\begin{definition}[Segment]
A \emph{segment} of a path $P$ is any subpath between two consecutive netpoints on $P$. For an edge $e \in P$, define $seg(e, P)$ as the unique segment in $P$ that contains $e$.
\end{definition}

In \cite{ChechikCFK17}, the authors proved the following lemma that bounds the length of the segment containing an edge $e$. 

\ifnum\myone=1 We state the lemma below and, for the sake of completeness, prove it in the appendix (See \Cref{Appendix_A})
\else
See the full version of the paper below for the proof of this lemma
\fi

\begin{lemma}[\cite{ChechikCFK17}, Lemma 2.1]
\label{lem:chechik}
    Let $P$ be a path from $s$ to $t$, and let $e = (u,v)$ be an edge of $P$ (where $u$ is closer to $s$ than $t$ on $P$).  Then either $\operatorname{seg}(e, P) = \{e\}$ or $|\operatorname{seg}(e, P)| \leq \ep \cdot \min\{|P[s,u]|,|P[u,t]|\} \le \ep |P|$.
\end{lemma}

The concept of a detour frequently appears in fault tolerant shortest path literature. Normally, detours are defined for paths that avoid some edges. In our paper, we will define detour for paths that have one extra property.  We begin by defining the traditional notion of a detour, starting with the following definition.

\begin{definition}[Prefix and Suffix of path $P$] 
Let $P$ be a (not necessarily shortest) path from $s$ to $t$. Let $e = (u,v)$ be an edge on $P$, where $u$ is closer to $s$ and $v$ is closer to $t$. Then $\pre(P,e)$ is the subpath $P[s,u]$ and $\suf(P,e)$ is the subpath $P[v,t]$.
\end{definition}

We now define the detour of a path avoiding a single edge and the start of the detour.

\begin{definition}\label{detour_of_a_path}(Detour of a path $R$)
    Let $R$ be a path from $s$ to $t$ that avoids an edge $e$ on the original shortest path $st$.
    Let $w$ be the last vertex of $R$ that also lies on $\pre(st,e)$, that is, the path $R[w,t]$ does not intersect with $\pre(st,e)$ except at $w$. Then the subpath $R[w,t]$ is the \emph{detour} of $R$, and we say the detour of $R$ \emph{starts} at $w$.
\end{definition} 

Now, we define special detours that follow some other properties.

\begin{definition}\label{single_edge_detour}(Detour of a path $R$ starting from a segment and joining below it)
Let $R$ be a path avoiding an edge $e$ on the $st$ path. Let $w$ be the last vertex of $R$ on $\pre(st,e)$. We say that the detour of $R$ starts from $w$ on the segment $seg(e,st)$ if:
\begin{enumerate}
    \item $w \in seg(e,st)$, and
    \item $R[w,t]$ does not intersect $seg(e,st)$ except at point $w$. (See \Cref{detour})
\end{enumerate}

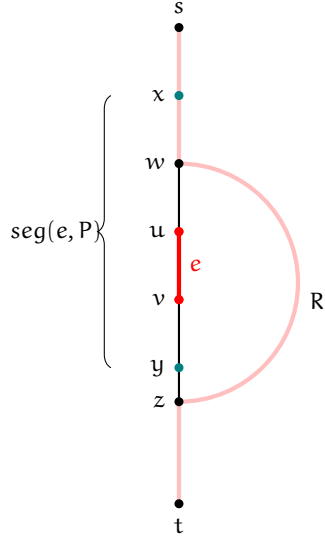
\begin{figure}[t]
    \centering
    \begin{tikzpicture}[scale=0.9, every node/.style={font=\small}]
        \tikzstyle{point} = [circle, fill=black, inner sep=1.2pt]
        \tikzstyle{mainpath} = [thick, black]
        \tikzstyle{detour} = [ultra thick, pink]
        \tikzstyle{edgefail} = [ultra thick, red]
        \tikzstyle{shortcut} = [orange, thick, dashed]

        \draw[mainpath] (0,7) -- (0,0);
        \draw[detour] (0,7)--(0,5);
        \draw[edgefail] (0,4)--(0,3) node[midway,right]{$e$};
        \draw[detour] (0,1.5) arc[start angle=-90,end angle=90,radius=1.75]node[thick, pos=0.5, below=7, sloped, black, rotate=-90] {$R$};
        \draw[detour] (0,1.5)--(0,0);
        \draw[decorate, decoration={brace, amplitude=5pt}]
        (-1,2) -- (-1,6) node[midway, left] {$seg(e,P)$};

        \node[point,label=left:$z$] at (0,1.5){};
        \node[point,red] at (0,4){};
        \node[point,red] at (0,3){};
        \node[point,label=above:$s$] at (0,7){};
        \node[point,red,label=left:$u$] at (0,4){};
        \node[point,red,label=left:$v$] at (0,3){};
        \node[point,label=left:$x$,teal] at (0,6){};
        \node[point,label=left:$w$] at (0,5){};
        \node[point,label=left:$y$,teal] at (0,2){};
        \node[point,label=below:$t$] at (0,0){};
    \end{tikzpicture}
    \caption{Let $P$ be a path from $s$ to $t$. Let $R$ be a path (shown in pink) avoiding $e$ that lies on the segment $xy$. Detour of $R$ is $R[w,t]$. The detour of $R$ starts on the segment $xy$ above $e$ at $w$ and joins the path at $z$ below  the segment $xy$.}
    \label{detour}

\end{figure}

Let $z$ be the first vertex of $R$ such that $z \in \suf(st, e)$, that is the path $R[s,z]$ does not intersect with $\suf(st,e)$ except at $z$. By the second condition above, $z$ does not lie on the segment $seg(e,st)$ and we say that the detour joins $st$ below $seg(e,st)$ at $z$.
\end{definition}

\section{Overview}
\label{sec:overview}

Consider the algorithm of \cite{ChechikCFK17}, which designs a $(\Ep)$-approximate $\ado(f)$ with size roughly $\tl(n^2)$ and query time $O(f^5)$. We focus on the oracle's space complexity. The algorithm builds a data structure $\ft(st)$ for every vertex pair $(s,t)$ in the graph. While computing the shortest $s$-$t$ path avoiding $F$, we query $\ft(st)$. This structure guarantees the following property.

\begin{lemma}[stated informally, omitting technical details]
Either $\ft(st)$ outputs $|st \dia F|$, or there exists an affected vertex $w \in V(F)$ {\bf close} to the path $st \dia F$. In the latter case, we can compute $|sw \dia F|$ exactly using $\ft(sw)$.
\end{lemma}

Note that the above lemma is only an informal summary of \cite{ChechikCFK17} and omits technical details for clarity. The lemma states that using $\ft(st)$, we either obtain $|st \dia F|$ directly or identify an affected vertex $w \in V(F)$ {\bf close} to the path $st \dia F$ (the notion of closeness is not formally defined here). We then compute $|sw \dia F|$ and $|wt \dia F|$. Since $w$ is close to $st \dia F$, their sum approximates $|st \dia F|$ well. The lemma ensures that $\ft(sw)$ returns $|sw \dia F|$ exactly. For $|wt \dia F|$, the source is now an affected vertex, so we recurse and compute it using $\ft(wt)$. The recursion depth is bounded by $2f$, the number of vertices in $V(F)$. In fact, \cite{ChechikCFK17} implement this process in $O(f^5)$ time.

We now focus on the space taken by the data-structure. In \cite{ChechikCFK17}, the authors constructed an $(\Ep)$-approximate $\ado(f)$ for all-pairs distance queries, which justified their $\tl(n^2)$ space. In contrast, our setting involves only a single source, where we aim to design a $\sdo(2)$ of size $o(n^2)$ with fast query time. Thus, we cannot construct an $\ft$ data structure for every vertex pair in the graph. At the same time, we require fast query time, and achieving both objectives simultaneously is challenging. To highlight this difficulty, we next examine a related problem.

In \cite{DBLP:journals/theoretics/BiloCCC0KS24} (and then in a subsequent improvement in \cite{BiloCCC0S24}), the authors design a $(3+\ep)$-approximate $\ado(f)$. Their oracle has size $o(n^2)$ but suffers from $\Omega(n^c)$ query time (for some constant $c < 1$).  They face the same challenge as ours: using the $\ft$ data structure requires quadratic space if built for all pairs. To avoid this, they sample a subset of vertices and construct $\ft$ data structures from each sampled vertex to all others. Since the sample size is much smaller than $n$, the overall space remains $o(n^2)$. However, this space-saving comes at the cost of $ O (n^c) $ query time. Reducing this query time remains an important open problem.

Our challenge closely resembles the above problem, even though we focus on the simpler single source setting. In fact, our goal is harder, as we seek polylogarithmic query time. The reader will observe that we introduce significant modifications to the $\ft$ data structure, enabling us to reduce the number of vertex pairs for which it needs to be constructed. To fully appreciate our new contribution, we first present a construction of $\sdo(1)$ that achieves $\tl(n\sqrt{n})$ space and $\tl(1)$ query time. Although this result is weaker than that of \cite{BaswanaK13}, who obtained nearly $\tl\!\left(\tfrac{n}{\ep^3}\right)$ space and $O(1)$ query time, it serves as an essential foundation for our improved dual fault structure.

	Our construction naturally extends to handle the case of two faults. Once the reader is familiar with the working of $\sdo(1)$, the extension to $\sdo(2)$ becomes straightforward. Thus, the next section acts as an overview of our approach.

\section{Warm-up: Single fault tolerant Distance Oracle}
In this section, we design a single fault tolerant distance oracle. The space taken by our distance oracle is $\tl(n\sqrt n)$ and its query time is $\tl(1)$. Let us first see the construction of the data-structure that will be used by this oracle.

\subsection{Constructing \texorpdfstring{$\ft(st)$}{ft(st)}}

We construct a data-structure $\ft(st)$ for each $t \in V$. We build $\ft(st)$ using the function $\mn$ defined in \Cref{alg:constructft}. This function takes three parameters: a graph $H$, a path $P$, and the $depth$ of the node. It creates a node with these inputs and recursively builds its subtree. To construct the full tree $\ft(st)$, we invoke $\mn(G, st, 0)$, which initializes the root node at level~0. Since \Cref{alg:constructft} only recurses when $\text{depth} < 1$, the tree has exactly two levels: the root at depth~0 and its children at depth~1.

\begin{algorithm}[t]
Make a node $\rr$ with path $\pr_\rr = P$ \\
\If{$depth  < 1$}
{
    \ForEach{segment $xy$ in $\pr_\rr$}
    {
        Let $P(xy)$ be the shortest 1-decomposable path $Q_1 \odot e_1 \odot Q_2$
        in the original graph $G$, where each $Q_j$ is a shortest path in $G$, such that no edge of 
        $P(xy)$ belongs to the segment $xy$ \\
        \If{$P(xy)$ exists}
        {
            a child of $\rr$ $\leftarrow \mn(H-xy, P(xy), depth+1)$\\
        }
        \ForEach{$i = 0$ to $\log_\Ep nW$}
        {
            Let $\mathcal{Q} = \{ \pr_\rr \dia e \mid e \in xy \}$
            \tcp{$\pr_\rr \dia e$ denotes the shortest path between 
            the endpoints of $\pr_\rr$ avoiding edge $e$, computed in $G$}

            Let $P_i$ be the path in $\mathcal{Q}$ whose detour starts 
            closest to $x$ on the segment $xy$ and detour joins below 
            $xy$ on $\pr_\rr$, and its detour length $\le (\Ep)^i$ 

            \If{$P_i$ exists}
            {
                a child of $\rr \leftarrow \mn(H, P_i, depth+1)$
            }
        }
    }
    
}
return $\rr$
\caption{$\mn(H, P, depth)$}
\label{alg:constructft}
\end{algorithm}
 In $\mn(G,st,0)$, we first make the root  of $\ft(st)$, that is $\ro(\ft(st))$. Each node of $\ft(st)$ is associated with a graph and a primary path $\pr_\rr$. For instance, for the root $\rr$ of $\ft(st)$, the associated graph is $G$, and the primary path $\pr_\rr$ is the $st$ path. For each segment $xy \in \pr_\rr$, we do the following:

  \begin{enumerate}
  
    \item Make a Segment child 

	Let $P(xy)$ be the shortest 1-decomposable path of the form $Q_1 \odot e_1 \odot Q_2$ in the original graph $G$, where each $Q_j$ is a shortest path in the original graph, such that no edge of $P(xy)$ belongs to the segment $xy$. In other words, $P(xy)$ is a 1-decomposable path in $G$ that survives in $G - xy$. If $P(xy)$ exists, then we create a child for the segment $xy$ by calling $\mn(G - xy, P(xy), 1)$ and designate it as a \emph{segment} child.

    \item Make Detour children

    Let $\mathcal{Q}$ be the set of all shortest paths between the endpoints of $\pr_\rr$ that avoid some edge of the segment $xy$, computed in $G$. Formally, $\mathcal{Q} = \{\pr_\rr \diamond e \mid e \in xy\}$. Let $P_i$ be the path in $\mathcal{Q}$ such that the detour starts closest to $x$ on the segment $xy$, the detour joins $\pr_\rr$ below $xy$ on $\pr_\rr$, and the detour length is $\le (\Ep)^i$. If $P_i$ exists, then we create a \emph{detour} child $\rr_i$ of $\rr$ with $P_i$ as its primary path by invoking $\mn(G, P_i, 1)$.

  \end{enumerate}
     
For each segment $xy$, we create one segment child and at most $\loge$ detour children of $\rr$. Since $\pr_\rr$ contains $O(\loge)$ segments, there are $O(\loge)$ segment children and $O(\logee)$ detour children at level 1 of $\ft(st)$. We thus claim the following:

\begin{lemma}
\label{nodes_in_FT}
	There are $O(\logee)$ nodes in $\ft(st)$.
\end{lemma}

We now try to bound the size of $\ft(s,t)$. To this end, we claim the following crucial lemma.

\begin{lemma}
\label{decompose_node}
For each node $\xx \in \ft(st)$, $\pr_\xx$ is 1-decomposable. 
\end{lemma}
\ifnum\myone=1
\begin{proof}
     The root $\rr$ of $\ft(st)$ stores the shortest $st$ path, which is 0-decomposable (a single shortest path) and hence also 1-decomposable. Since \Cref{alg:constructft} only recurses when $\text{depth} < 1$, the tree has exactly two levels: the root at depth~0 and its children at depth~1. It suffices to verify that each child's primary path is 1-decomposable.

	\noindent\textit{Segment children.} For a segment $xy \in \pr_\rr$, the path $P(xy)$ is defined as the shortest 1-decomposable path in the original graph $G$ (that also survives $G-xy$), so 1-decomposability holds by construction.

	\noindent\textit{Detour children.} The path $P_i$ lies in $\mathcal{Q} = \{st \diamond e \mid e \in xy\}$, meaning it is a shortest $st$ path in $G$ avoiding a single edge $e \in xy$. By \Cref{lem:decompasable}, every shortest path after one edge failure is 1-decomposable, so $P_i$ is 1-decomposable.

	Hence, all nodes of $\ft(st)$ store a 1-decomposable primary path.
\end{proof}
\else
    See the full version of the paper below for the proof.
\fi

We now use the above lemma to store the primary path at each node of $\ft(st)$ efficiently. Consider a node $\rr$ of $\ft(st)$. The lemma implies that $\pr_\rr$ is $1$ decomposable, so it has the form $Q_1 \odot e \odot Q_2$. If a path $Q_j$ has  at most $\sqrt{n}$ edges, we store it explicitly at node $\rr$. If $Q_j$ has more than  $\sqrt{n}$ edges, we apply a sampling approach to represent $Q_j$. A similar 
technique has been used in \cite{DBLP:journals/theoretics/BiloCCC0KS24}. We 
sample a set $L$ of nodes from $G$ independently with probability 
$\tl\!\left(\frac{1}{\sqrt{n}}\right)$. With high probability, 
$|L| = \tilde{O}(\sqrt{n})$. We store the shortest path tree rooted at each 
node in $L$, which requires total space $\tilde{O}(n\sqrt{n})$. Now consider a subpath $Q_j$ from $a$ to $b$. Since  $Q_j$ has $\ge \sqrt{n}$ edges, a vertex 
$z \in L$ lies on $Q_j$ with high probability. Moreover, since subpaths of 
shortest paths are themselves shortest paths, both $Q_j[a,z]$ 
and $Q_j[z,b]$ are shortest paths in $G$. We therefore store $Q_j$ implicitly 
by recording only the pairs $(a, z)$ and $(z, b)$ at node $\rr$, and recover the full 
paths $Q_j[a,z]$ and $Q_j[z,b]$ from the shortest path tree $T_z$ rooted at 
$z \in L$. Thus, each node $\rr$ stores $Q_j$ in $O(1)$ space when 
$Q_j$ has $\ge \sqrt{n}$ edges.

Combining both cases: each node $\rr$ represents its primary path in 
$O(\sqrt{n})$ space (either the explicit path of length at most $\sqrt{n}$, or 
an $O(1)$-sized pairs $(a,z)$ and $(z,b)$). Since $\ft(st)$ has $O(\logee)$ nodes by 
\Cref{nodes_in_FT}, each tree occupies $\tl(\sqrt{n})$ space. Across all $n$ 
trees, the total space is $\tl(n\sqrt{n})$. Also, we build shortest path trees from all vertices of $L$. This also takes $\tl(n\sqrt n)$ space. Thus, the total space taken by our data-structure is $\tl(n\sqrt n)$.

In our query algorithm, we must efficiently check whether a given edge $e$ lies on the path $Q_j$. If $Q_j$ has $\le \sqrt{n}$ edges, we store a dictionary at node $\rr$ using $\tilde{O}(\sqrt{n})$ space, allowing constant-time membership queries for $e \in Q_j$. If $Q_j$ has  $> \sqrt{n}$ edges, we represent $Q_j$ implicitly as the concatenation of two shortest paths, $Q_j[a,z]$ and $Q_j[z,b]$, where $z \in L$ lies on $Q_j$. Now we use the following lemma to check if $e$ lies in $Q_j[a,z]$.

\begin{lemma}
    \label{lem:lca}(See \cite{BenderF00} and its references) Given a tree $T$ with $n$ vertices, we can construct a data structure of size $O(n)$ in $O(n)$ time, allowing us to answer LCA(Lowest Common Ancestor) queries in $O(1)$ time.
\end{lemma}

As we are storing the path $Q_j[a,b]$ as a pairs $(a,z)$ and $(z,b)$ where $z \in L$ , we will proceed as follows. For a faulty edge $e=(u,v)$, we first check if $T_z$ contains $e$. If $T_z$ does not contain $e$, then $Q_j$ cannot contain $e$. If $T_z$ contains $e$, then we check whether $LCA(a,u)=u$ and $LCA(a,v)=v$ are satisfied or not. If both conditions hold, then the edge $e$ is in the shortest path $az$. These LCA queries take $O(1)$ time. A similar check applies for $zb$. With an additional $\tilde{O}(1)$ data structure at $\rr$, we can also report which segment of $Q_j$ contains $e$.
\ifnum\myone=1 We  prove this later (see \Cref{Finding_the_segment_of_an_edge}),
\else
See the full version of the paper below for the proof.
\fi

\begin{lemma}\label{lem:insegment}
    For each $t$ , each node of  $\ft(st)$ can be represented using $\tl(\sqrt n)$ space. Thus, the total space taken by $\ft(st)$ overall $t \in V \setminus \{s\}$ is $\tl(n\sqrt n)$. Moreover, the time taken to determine if $e \in \pr_\rr$ in $\tl(1)$ for each node $\rr$ in $\ft(st)$. Moreover, we can also report the segment of $\pr_\rr$ that contains $e$ in $\tl(1)$ time.
\end{lemma}

\subsection{Query Algorithm}
Given a destination vertex $t$ and a faulty edge $e = (u,v)$, we execute the algorithm $\qu(t, \ro(\ft(st)))$; assuming $F = \{e\}$ is implicitly known. Our algorithm also requires one additional implicit array, which we describe next. For brevity, we omit it from the procedure's arguments.

The procedure $\qu$ is recursive. The initial call uses $t$ as the first parameter (the destination), while subsequent calls use one of the affected vertices 
\ifnum\myone=1
	(see \Cref{query:11}).
\else
	(see Line \ref{query:11}).
\fi
 The first \textbf{if} condition bounds the recursion depth.  During $\qu(t,\cdot)$, we track the recursion path using an array that records all vertices visited along the current path. The array has constant size because, after the root, each recursive call involves only affected vertices. 
If $\qu$ encounters a vertex that already appears as the first parameter on the current path, it returns $\infty$.

Let $\rr$ be the root of $\ft(st)$. Thus,  $\pr_\rr = st$. If $st$ contains no faulty edge, we return it directly. Otherwise, $e=(u,v)$ lies on the $st$ path, and let $xy$ denote the segment of $st$ that contains $e$. The algorithm then proceeds through three main cases depending on the structure of $st \dia e$. We now go through these three cases.

	\begin{algorithm}[t]
	\DontPrintSemicolon
	\If{any ancestor function call of $\qu$ had $t$ as the first parameter}
	{
		return $\infty$
	}
	\If{$\pr_\rr$ does not contain $e$}{
		\Return $|\pr_\rr|$
	}

	\Else{
		$\pa \leftarrow \infty$
		
		Let $e = (u,v)$ be the faulty edge on $\pr_\rr$, and let $xy \in \pr_\rr$ be the segment containing $e$\\
        \label{query:7}
		\Comment{\textcolor{blue}{ Part 1: $st \dia e$ avoids segment $xy$ }}
		Let $\xx$ be the segment child of $\rr$ corresponding to the path avoiding $xy$\\
		\If{$\xx$ exists}
		{
			$\pa \leftarrow |\pr_\xx|$
		\label{query:9}
		}
		
		\vspace{2mm}
		\Comment{\textcolor{blue}{ Part 2: $st \dia e$ intersects segment $xy$ below $e$ }}
		\If{$v \neq t$}
		{
		
		$\pa \leftarrow \min\{\pa,  \qu(v, \ro(\ft(sv))) + |\pr_\rr[v,t]|\}$
		\label{query:11}
		}
		\vspace{2mm}
		\Comment{\textcolor{blue}{ Part 3: $st \dia e$ intersect segment $xy$ above $e$}}
		\ForEach{$i = 0$ to $\log_\Ep nW$ \label{query:12}}{
			Let $\rr_i$ be the detour child of $\rr$ whose detour in $\pr_{\rr_i}$ starts on $xy$ and joins below $y$ on $\pr_\rr$ and has length at most $(\Ep)^i$
			
			\If{$\rr_i$ exists and $\pr_{\rr_i}$ avoids $e$}
			{
				$\pa \leftarrow \min\{\pa, |\pr_{\rr_i}| \}$
				\label{query:15}
			}

			}
		}
	\Return $\pa$
	\caption{\qu($t, \rr$)}
	\label{query}
\end{algorithm}

	\begin{enumerate}
		\item Part 1: $st \dia e$ avoids the segment $xy$ (See part 1 of \Cref{fig:single_edge_fault})

		Let $P(xy)$ be the shortest 1-decomposable path in $G$ that survives in $G - xy$ constructed during preprocessing.  Let $\xx$ be the segment child that contains $P(xy)$. If $\xx$ exists, then we set $\pa$ to $|\pr_\xx|$ = $|P(xy)|$.
		
		\item $st \dia e$ intersect $xy$.
		
		There are two cases here, depending on whether $st \dia e$ intersects $xy$ above or below $e$. If it intersects $xy$ both above and below $e$, we will give preference to its intersection below $e$ in the analysis. Let's go over these two cases.
		
		\begin{enumerate}
		\item\label{case:qn2a} Part 2: $st \dia e$ intersects $xy$ below $e$. (See part 2 of \Cref{fig:single_edge_fault})

		Let $v$ be the endpoint of $e$ closest to $t$ on the $\pr_\rr = st$ path. We estimate $|st \dia e|$ by computing $|\pr_\rr[v,t]| + |sv \dia e|$. Note that for a single fault, the first path is just the shortest path $vt$.  Although $|sv \dia e|$ is unknown, we can approximate it by calling $\qu(v, \ro(\ft(sv)))$. One may question whether this recursive call makes any progress. It does - instead of querying at $t$, we now query at $v$, an affected vertex, thereby making some progress.
				
		\item Part 3: $st \dia e$ intersects $xy$ above $e$. (See part 3 of \Cref{fig:single_edge_fault})

		It means that $st \dia e$ does not intersect $xy$ below $e$. This implies that the detour of $st \dia e$ begins on the segment $xy$ above $e$ and joins below $xy$ on $st$; a hard case for our algorithm. We had precomputed detour children of the root precisely for such situations, and now we utilise them. For each $i \in [0, \loge]$, we examine the detour child $\rr_i$ of $\rr$ such that the detour of $\rr_i$ starts on $xy$ and has length at most $(\Ep)^i$.  If  $\pr_\rri$ avoids $e$, we update $\pa$ to $|\pr_\rri|$.

		\end{enumerate}
	\end{enumerate}
	
\begin{figure}
    \centering
    \begin{minipage}[t]{0.22\textwidth}
        \centering
        \begin{tikzpicture}[scale=0.8, every node/.style={font=\small}]
            \tikzstyle{point} = [circle, fill=black, inner sep=1.5pt]
            \tikzstyle{fault} = [red, line width=0.8mm]

            \draw[thick] (0,7) -- (0,0);
            
            \draw [pink, ultra thick, opacity=0.7] (0,6) arc [start angle=-90, end angle=90, radius=-2.5] node[pos=0.5, below=14, sloped, black, rotate=180] {$st \dia e$}; 
            \draw [pink, ultra thick, opacity=0.7] (0,6) -- (0,7);  
            \draw [pink, ultra thick, opacity=0.7] (0,1) -- (0,0);

            \node[point, label=right:$z$] at (0,1) {};
            \draw[fault] (0,3) -- (0,4);
            \draw (0,3) -- (0,4) node[midway, left] {$e$};
            \draw node at (0,-1.1) {Part 1};
            \draw (-0.7,-0.7) rectangle (0.7,-1.5);
             \draw node at (0,-2) {};
             \node[point, label=above:$s$] at (0,7) {};
            \node[point, label=right:$x$, teal] at (0,5) {};
            \node[point, red, label=right:$v$] at (0,3) {};
            \node[point, red, label=right:$u$] at (0,4) {};
            \node[point, label=right:$y$, teal] at (0,2) {};
            \node[point, label=below:$t$] at (0,0) {};

        \end{tikzpicture}
        
    \end{minipage}
        \begin{minipage}[t]{0.22\textwidth}
        \centering
        \begin{tikzpicture}[scale=0.8, every node/.style={font=\small}]
            \tikzstyle{point} = [circle, fill=black, inner sep=1.5pt]
            \tikzstyle{fault} = [red, line width=0.8mm]

            \draw[thick] (0,7) -- (0,0);
            
            \draw [pink, ultra thick, opacity=0.7] (0,6) arc [start angle=-90, end angle=90, radius=-2] node[pos=0.5, below=14, sloped, black, rotate=180] {$st \dia e$}; 
            \draw [pink, line width=0.8mm, opacity=0.7] (0,6) -- (0,7);  
            \draw [pink, line width=0.8mm, opacity=0.7] (0,2) -- (0,0);

            \draw[fault] (0,3) -- (0,4);
            \node[point, label=right:$z$] at (0,2) {};
            \draw (0,3) -- (0,4) node[midway, left] {$e$};
           \draw (0,-0.7) rectangle (5,-1.5);
  			\node at (2.4,-1.1) {Part 2};
  		\draw node at (0,-2) {};
        \node[point, label=right:$y$, teal] at (0,1) {};
        \node[point, label=above:$s$] at (0,7) {};
        \node[point, label=right:$x$, teal] at (0,5) {};
        \node[point, red, label=right:$v$] at (0,3) {};
        \node[point, red, label=right:$u$] at (0,4) {};
        \node[point, label=below:$t$] at (0,0) {};

        \end{tikzpicture}

    \end{minipage}
    \begin{minipage}[t]{0.22\textwidth}
        \centering
        \begin{tikzpicture}[scale=0.8, every node/.style={font=\small}]
            \tikzstyle{point} = [circle, fill=black, inner sep=1.5pt]
            \tikzstyle{fault} = [red, line width=0.8mm]

            \draw[thick] (0,7) -- (0,0);

            \draw [pink, ultra thick, opacity=0.7] (0,4.5) arc [start angle=-90, end angle=90, radius=-1.25] node[pos=0.5, below=14, sloped, black, rotate=180] {$st \dia e$}; 
            \draw [pink, line width=0.8mm, opacity=0.7] (0,4.5) -- (0,7);  
            \draw [pink, line width=0.8mm, opacity=0.7] (0,2) -- (0,0);

            \draw[fault] (0,3) -- (0,4);
            \draw (0,3) -- (0,4) node[midway, left] {$e$};
            \node[point, label=right:$z$] at (0,2) {};
  
            \draw node at (0,-2) {};
            \node[point, label=right:$x$, teal] at (0,5) {};
            \node[point, label=right:$y$, teal] at (0,1) {};
            \node[point, label=above:$s$] at (0,7) {};
            \node[point, red, label=right:$v$] at (0,3) {};
            \node[point, red, label=right:$u$] at (0,4) {};
            \node[point, label=below:$t$] at (0,0) {};

        \end{tikzpicture}

    \end{minipage}
	\begin{minipage}[t]{0.22\textwidth}
        \centering
        \begin{tikzpicture}[scale=0.8, every node/.style={font=\small}]
            \tikzstyle{point} = [circle, fill=black, inner sep=1.5pt]
            \tikzstyle{fault} = [red, line width=0.8mm]

            \draw[thick] (0,7) -- (0,0);
            
            \draw [pink, ultra thick, opacity=0.7] (0,5) arc [start angle=-90, end angle=90, radius=-1.5] node[pos=0.5, below=14, sloped, black, rotate=180] {$st \dia e$};

            \draw [blue, ultra thick, opacity=0.7] (0,1.5) arc [start angle=-90, end angle=90, radius=2] node[pos=0.5, below=14, sloped, black, rotate=180] {$\pr_\rri$};
            
            \draw [pink, line width=0.8mm, opacity=0.7] (0,5) -- (0,7);  
            \draw [pink, line width=0.8mm, opacity=0.7] (0,1.5) -- (0,0);

            \draw[fault] (0,3) -- (0,4);
            \draw (0,3) -- (0,4) node[midway, left] {$e$};
            \draw node at (0,-1.1) {Part 3};
            \draw (-0.7,-0.7) rectangle (0.7,-1.5);
            \draw node at (0,-2) {};
            \node[point, blue, label=left:$w$] at (0,5.5) {};
            \node[point, label=right:$x$, teal] at (0,6) {};
             \node[point, label=above:$s$] at (0,7) {};
            \node[point, red, label=right:$v$] at (0,3) {};
            \node[point, red, label=right:$u$] at (0,4) {};
            \node[point, label=right:$y$, teal] at (0,2.3) {};
            \node[point, label=below:$t$] at (0,0) {};
            \node[point, label=right:$z$] at (0,5) {};

        \end{tikzpicture}

    \end{minipage}

    \caption{ Part 1: $st \dia e$ avoids segment $xy$. Part 2: $st \dia e$ intersects $xy$ below $e$ on segment $xy$. Part 3: $st \dia e$ intersects $xy$ above $e$ on segment $xy$.
    }
    \label{fig:single_edge_fault}
\end{figure}
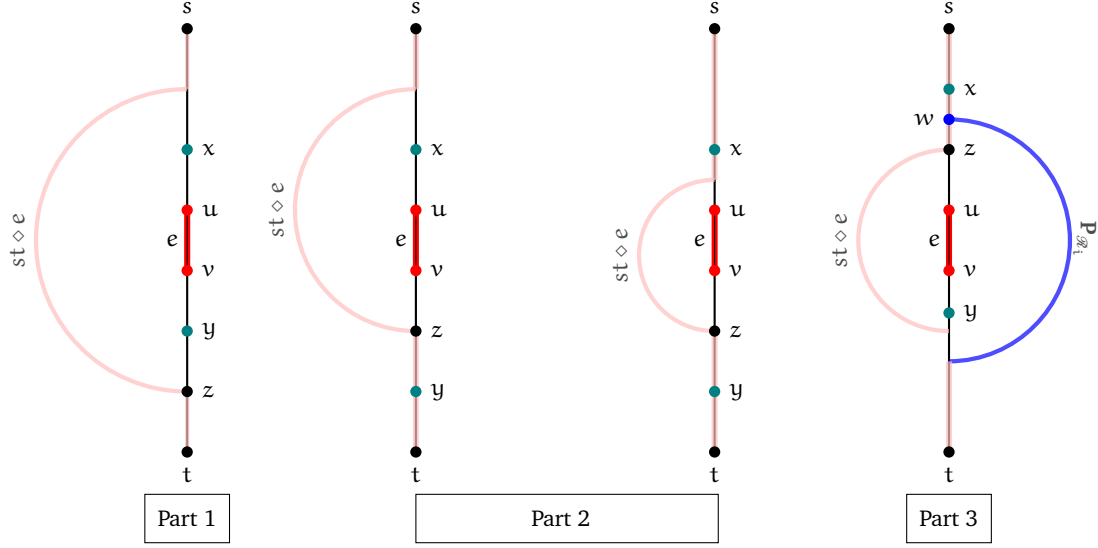

\subsection{Running Time}
\label{sec:onefaultrunning}

In this section, we analyze the running time of our $\qu$ algorithm. To this end, we introduce a simple method that will also be useful for handling the dual-fault case. We represent the execution of $\qu$ using a \emph{recursion tree}, where each node corresponds to a call to $\qu$. The root is $\qu(t,\ro(\ft(st)))$, and its children are the recursive calls triggered by this invocation. Recursively expanding the children yields the full tree, which succinctly captures the running time of the algorithm. 

The running time of $\qu$ per node is dominated by the loop in Part 3
\ifnum\myone=1
	(see \Cref{query:12}.
\else
	(see Line \ref{query:12}).
\fi 
taking $O(\log_\Ep nW)$ time, and it makes at most one recursive call. Consequently, in our case, the recursion tree degenerates to a path. We will show that this path has length at most 3, resulting in a total of 4 nodes. Hence, the overall running time of the algorithm is $O(\log_{1+\ep} nW)$.

In the recursion tree of $\qu$, every call after the root uses an affected vertex as its first parameter. Since there are only two affected vertices, any attempt to reuse an affected vertex terminates immediately, returning $\infty$. Consequently, the path in the recursion tree can have at most three recursive calls. Thus, we can claim the following lemma.

\begin{lemma}
	The running time of $\qu$ is $O(\log_{1+\ep}nW)$.
\end{lemma}

We now proceed to prove the correctness of our algorithm. To this end, we first establish a useful property of our algorithm.

\subsection{Property of our algorithm}

Our algorithm addresses two main challenges: (1) when $\stf$ intersects a segment containing a faulty edge below $e$, and (2) when the detour of $\stf$ starts above $e$ on the same segment. We first present a simple lemma to handle the first case. This lemma is general-purpose, applying even in the presence of two faults, and will be used extensively in the analysis of the dual fault case.

\begin{lemma}
\label{lem:belowe}
	  Let $\rr$ be a node in $\ft(st)$.
	  Let $e= (u,v)$ be the {\bf last} faulty edge on the path $\pr_\rr$ with $u$ closer to $s$. Let $S = seg(e,\pr_\rr)$.  Suppose the fault tolerant path $R = \stf$ intersects $S$ at a vertex $z$ below $e$. Assume there exists a path $Q$ from $z$ to $t$ that also avoids $F$. If the following conditions hold:	
\begin{enumerate}[noitemsep]
	\item $\qu(v, \ro(\ft(sv)))$ returns an $\alpha$-approximation of $|sv \dia F|$,
	\item $|Q|$ is an $\alpha$-approximation of $|R[z,t]|$, and
	\item $|S|$ is a $\beta$-approximation of $|R|$ where $0\le \beta \le 1$ 
\end{enumerate}

then the path of length $|Q| + |S[z,v]| + \qu(v, \ro(\ft(sv)))$ gives an $(\alpha + \beta + \alpha\beta)$-approximation of $|st \diamond F|$.
\end{lemma}
\ifnum\myone=1

\begin{proof}

 Let $\pa = |Q| + |S[z,v]| +\qu(v,\ro(\ft(sv)))$. We now bound the $\pa$ as follows:

		 \begin{align}
			\pa &=    |Q| + |S[z,v]| +\qu(v,\ro(\ft(sv)))\notag\\
			\intertext{Using the condition of the lemma, $\qu(v,\ro(\ft(sv)))$ returns $\al |sv \dia F|$. Thus, we get:}
			 	& \le |Q| + |S[z,v]| + \al |sv \dia F|\notag\\
			 	\intertext{Using triangle inequality, we get:}
			 	& \le |Q| + |S[z,v]|+ \al (|sz \dia F|+|zv \dia F|)\notag\\
			 	\intertext{$S[z,v]$ also avoids $F$ as $e$ is the last edge on the path $\pr_\rr$. Thus,  $|zv \dia F| \le |S[z,v]|$. Adding it to the above equation, we get:  }
			 	& \le |Q| + |S[z,v]| + \al(|sz \dia F| + |S[z,v]|)\notag\\
			 	\intertext{Since $S[z,v]$ is a sub-segment of $S$, its length is at most $|S|$. Thus: }
			 	& \le |Q| + (1+\al) |S| + \al|sz \dia F| \notag\\
			 	\intertext{Since $R$ is a  fault-tolerant path from $s$ to $t$ passing through $z$, we have $|sz \dia F| = |R[s,z]|$. Using condition (2) $|Q| \le \al |R[z,t]|$ and condition (3) $|S| \le \be |R|$ of the lemma, we get: }
			 	& \le \al |R[z,t]| + (1+\al)\be|R|  + \al|R[s,z]|\notag\\
			 	& = (1+\al)\be|R| + \al(|R[z,t]| + |R[s,z]|)\notag\\
			 	& = (1+\al)\be|R| + \al|R|\notag\\
			 	& = (\al+\be+\al\be)|R| \notag
		\end{align}
This completes the proof of the lemma.	
\end{proof}
\else
	See the full version of the paper below for the proof.
\fi
We now show how to apply the lemma. Suppose we are processing the root $\rr$ of $\ft(st)$, and edge $e$ lies on the segment $S= xy$ of the path $\pr_\rr = st$. Assume that $st \dia e$ intersects $xy$ below $e$ at the vertex $z$. In
\ifnum\myone=1
    \Cref{query:11}
\else
    Line \ref{query:11}
\fi 
we compute
\begin{align}
	\pa &= |\pr_\rr[v,t]|+\qu(v,\ro(\ft(sv)))\notag\\
	\intertext{Since $z$ lies between $v$ and $t$ on that path $\pr_\rr$, we get:}
	&= |\pr_\rr[z,t]| + |\pr_\rr[z,v]| + \qu(v, \ro(\ft(sv)))\notag\\
	\intertext{Since $\pr_\rr[z,v]$ lies entirely in the segment $S$, we have $\pr_\rr[v,z] = S[z,v]$ }
	& = |\pr_\rr[z,t]| + |S[z,v]| +  \qu(v, \ro(\ft(sv)))\notag\\
	\intertext{Setting $Q = \pr_\rr[z,t]$, we get}
	& = |Q| + |S[z,v]| +  \qu(v, \ro(\ft(sv))) \label{eq:one} 
\end{align}

Thus, the path computed by our algorithm in 
\ifnum\myone=1
    \Cref{query:11}
\else
    Line \ref{query:11}
\fi
is the same as the path described in \Cref{lem:belowe}.

\subsection{Correctness}

Consider finding the shortest path from $s$ to $t$ that avoids the edge $e = (u,v)$. 
In the function $\qu(t,\ro(\ft(st)))$, the first parameter is initially $t$, and in subsequent recursive calls it may be an affected vertex. 
Order $t$ and the affected vertices by their distance from $s$ in the graph $G-e$, and let $L$ denote the subsequence of all vertices up to $t$ in this order. 
  If $t$ is the closest vertex to $s$ in $G-e$, $L$ contains only $t$.  In the worst case, the size of $L$ is 3. 
We now prove the following lemma.

\begin{lemma}
\label{lem:k}
Let $p$ be the $k$-th vertex in the sequence $L$ where $k \le |L|$. Then $\qu(p, \ro(\ft(sp)))$ returns a $(1 + k^3\ep)$-approximation of $|sp \dia e|$.
\end{lemma}

The above lemma implies that $\qu(t,\ro(\ft(st)))$ returns a $(1+27\ep)$ approximation of $|st \dia e|$ as $|L| \le 3$. We prove the lemma by induction on the vertices of $L$. For the base case ($k=1$), without loss of generality, let $u$ be the first vertex in the sequence. The shortest $su$ path avoids the faulty edge $e$, and $\qu(u, \ro(\ft(su)))$ returns $|su|$ exactly, establishing the base case. For the induction step, assume that there are $k-1$ vertices before $p$ in  $L$. We now prove it for the $k$-th vertex $p$. For notational convenience, we relabel $p$ as $t$, as this notation aligns with all our lemmas and discussions above.

 If the faulty edge $e$ does not lie on the shortest path $\pr_\rr$ (where $\rr$ is the root of  $\ft(st)$), $\qu(t, \ro(\ft(st)))$ simply returns the length $|\pr_\rr|$. Hence, assume that $e = (u,v)$ lies on $\pr_\rr$ within the segment $xy$, where $u$ is closer to $s$ than $v$. Let $R = st \dia e$. We first consider the simpler case where $R$ also avoids the segment $xy$.

Let $\xx$ be the segment child of $\rr$ corresponding to the removal of the segment $xy$. The segment child $\xx$ must exist as $R$ itself is one candidate for $\pr_\xx$. Since $\pr_\xx$ avoids all edges of $xy$ and $e \in xy$, it avoids $e$. Moreover, $\pr_\xx$ is a valid $s$ to $t$ path avoiding $e$, so $|\pr_\xx| \ge |R|$. Conversely, $R$ itself is a candidate 1-decomposable path avoiding $xy$ (since $R$ avoids $xy$ entirely) and $\pr_\xx$ is the shortest 1-decomposable path avoiding $xy$. So, $|\pr_\xx| \le |R|$. Thus $ |\pr_\xx| = |R|$. Consequently, in  
\ifnum\myone=1
    \Cref{query:9}
\else
    Line \ref{query:9}
\fi
we set $\pa = |\pr_\xx|$. 

We now analyze the case where $R$ intersects $xy$. If $R$ intersects $xy$ both above and below $e$, we will give preference to its intersection below $e$ in the analysis. The two possibilities are stated below:

\subsubsection{\texorpdfstring{$R$ intersects $xy$ above $e$ at a vertex $z$}{R intersects xy above e at a vertex z}}
\label{item:R intersect above e}

        It means that $R$ does not intersect $xy$ below $e$. Assume that the detour of $R$ starts at vertex $z$, and that $|R[z,t]| \in \left[(\Ep)^{i-1}, (\Ep)^i\right]$ for some $i \ge 1$. There exists a detour child $\rr_i$ of $\rr$ such that the detour of  $\pr_\rri$ starts above $e$, say at the vertex $w$  and $\pr_\rri[w,t]$ has length at most $(1+\ep)^i$.  Since $R$ itself is a candidate for $\pr_\rri$, such a path must exist.  In our construction, we require  the detour of $\pr_{\rr_i}$ to start as close to $x$ as possible on the segment $xy$. Hence, $w$ lies on or above $z$ on $\pr_{\rr}$. In our algorithm, we set $\pa$ to:
        	
        	\begin{align*}
        		\pa &= |\pr_\rri|\\
        			& =|\pr_\rri[s,w]| + |\pr_\rri[w,t]|
        			\intertext{But the detour of $\pr_\rri$ is of length at most $(\Ep)^i$. Thus, $|\pr_\rri[w,t]| \le (\Ep)^i \le (\Ep)|R[z,t]|$. Also, since the detour of $\pr_{\rr_i}$ starts at $w$ on $\pr_\rr$, $\pr_\rri[s,w] = \pr_\rr[s,w]$. Substituting these two terms, we get:}
        			& \le (\Ep) |R[z,t]| + |\pr_\rr[s,w]|
        			\intertext{But $|\pr_\rr[s,w]| \le |\pr_\rr[s,z]|$ as $w$ lies at or above $z$ on $\pr_\rr$. Thus, we get:}
        			& \le (\Ep) |R[z,t]| + |\pr_\rr[s,z]|\\
       			\intertext{But $\pr_\rr[s,z] = R[s,z]$ as the detour of $R$ starts at $z$ (See Part 3 of \Cref{fig:single_edge_fault}). Thus, we get:}
        			& \le (\Ep) |R[z,t]| + |R[s,z]|\\
        			& \le (\Ep) |R|
        	\end{align*}

\subsubsection{\texorpdfstring{$R$ intersects $xy$ below $e$, say at a vertex $z$}{R intersects xy below e at a vertex z}}
\label{item:R_intersect_below_e}
We first show that $|sv \dia e| < |st \dia e|$, implying $v \in L$. Since $st \dia e$ passes through a vertex $z$ on the segment $xy$, and $z$ lies closer to $v$ than to $t$, we expect $|sv \dia e| < |st \dia e|$. We now prove this formally with a general purpose lemma.

\begin{lemma}
\label{lem:vless}
Let $e =(u,v)$ be the edge closest to $y$ on the segment $xy$. Let $u$ be closer to $x$ than $y$ on the segment $xy$.
$R = st \dia e$ intersects $xy$ below $e$, say at vertex $z$. Then $|sv \dia e| \le |st \dia e|$. Similarly, if $R$ intersects $xy$ above $e$, say at $w$, and the subpath $wu$ of $xy$ does not contain any faults, then  $|su \dia F| \le |st \dia F|$.

\end{lemma}

\ifnum\myone=1
\begin{proof}
By the triangle inequality,
\begin{align*}
|sv \dia e| &\le |sz \dia e| + |zv \dia e|\\
\intertext{Since $z$ lies on the segment $xy$, and there is no faulty edge between $z$ and $v$, we can write $|zv \dia e| =|zv|$}
&\le |sz \dia e| + |zv|\\
\end{align*}

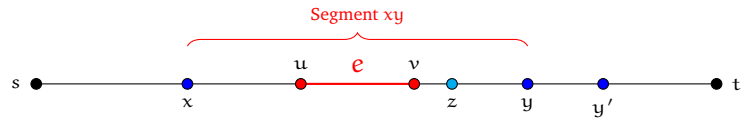
\begin{figure}[ht]
\centering
\begin{tikzpicture}[
    point/.style={circle,draw,fill=white,minimum size=4pt,inner sep=0pt},
    line/.style={thick},
    redpoint/.style={point,fill=red},
    bluepoint/.style={point,fill=blue},
    cyanpoint/.style={point,fill=cyan},
    every label/.style={font=\scriptsize},
    >=Stealth
]


\draw[thin] (0,0) -- (3.5,0);
\draw[thin] (5,0) -- (9,0);

\draw[line,red] (3.5,0) -- (5,0) node[midway, above=1pt, red] {$e$};

\draw [decorate,decoration={brace,amplitude=5pt},red]
  (2,0.5) -- (6.5,0.5) node[midway,yshift=0.4cm,red] {\scalebox{0.7}{Segment $xy$}};
  
\node[cyanpoint,label=below:$z$] at (5.5,0) {};
\node[point,black,label=left:$s$] at (0,0) {};
\node[bluepoint,label=below:$x$] (xp)at (2,0) {};
\node[redpoint,label=above:$u$] at (3.5,0) {};
\node[redpoint,label=above:$v$] at (5,0) {};
\node[bluepoint,label=below:$y$] (xq) at (6.5,0) {};
\node[bluepoint,label=below:$y'$] (xqq) at (7.5,0) {};
\node[point,black,label=right:$t$] at (9,0) {};

\end{tikzpicture}
\caption{Proving $|zv|\leq |tz|$.}

\label{fig:segment_prove}
\end{figure}

We now show that $|zv| \le |tz|$ (See \Cref{fig:segment_prove}). Suppose, for contradiction, that $|zv| > |tz|$. Let $y'$ be the net-point closest to $y$ originating from $t$ in the path $ty$. Thus, $ (\Ep)^{i+1} >  |ty'| \ge (\Ep)^i$ for some $i$. Note that $y$ may be equal to $y'$.  Thus, 
\begin{align*}
|vt|
&= |tz| + |zv| \\
&> 2|tz| \hspace{2cm} (\text{By assumption, $|zv| > |tz|$} )\\
\intertext{Since the path $tz$ contains $y'$, $|tz| > |ty'| \ge  (\Ep)^i$. Thus, we get:}
&\ |vt|\ge 2(1+\ep)^i \\
&\ge (1+\ep)^{i+1} 
\end{align*}
This implies that there must be a net-point originating from $t$ in $y'v$ path. But there is no net point originating from $t$ along the path $y'y$ (by construction) and between $yv$ (since $yv$ is a subpath of segment $xy$).  Thus, we arrive at a contradiction. Hence, $|zv| \le |tz|$. We now bound $sv \dia e$.

\begin{align*}
|sv \dia e|
&\le |sz \dia e| + |zv| \\
\intertext{For our discussion above, $|zv| \le |zt|$. Thus, we get:}
&\le |sz \dia e| + |zt| \\
&\le  |sz \dia e| + |zt \dia e| \\
&= |st \dia e| = |R| \hspace{2cm}(\text{Since $R$ passes through $z$})
\end{align*}

\begin{figure}[ht]
\centering
\begin{tikzpicture}[
    point/.style={circle,draw,fill=white,minimum size=4pt,inner sep=0pt},
    line/.style={thick},
    redpoint/.style={point,fill=red},
    bluepoint/.style={point,fill=blue},
    cyanpoint/.style={point,fill=cyan},
    every label/.style={font=\scriptsize},
    >=Stealth
]



\draw[thin] (0,0) -- (3.5,0);
\draw[thin] (5,0) -- (9,0);

\draw[line,red] (3.5,0) -- (5,0) node[midway, above=1pt, red] {$e$};

\draw[orange, ultra thick, dashed] 
(2.5,0) arc[start angle=180,end angle=0,radius=2.5]
node[midway, above] {$st \dia e$};

\draw [decorate,decoration={brace,amplitude=5pt},red]
  (2,0.5) -- (6.5,0.5) node[midway,yshift=0.4cm,red] {\scalebox{0.7}{Segment $xy$}};

\node[point,black,label=left:$s$] at (0,0) {};
\node[bluepoint,label=below:$x$] (xp)at (2,0) {};
\node[redpoint,label=below:$u$] at (3.5,0) {};
\node[redpoint,label=below:$v$] at (5,0) {};
\node[cyanpoint,label=below:$w$] at (2.5,0) {};
\node[bluepoint,label=below:$y$] (xq) at (6.5,0) {};
\node[point,black,label=right:$t$] at (9,0) {};

\end{tikzpicture}
\caption{$|su \dia e| \leq |st \dia e|$}

\label{fig:su_less_st}
\end{figure}

For the second part of the lemma, using triangle inequality, we know that , 

\begin{align}
|su \dia F| 
&\le |sw \dia F| + |wu \dia F| \label{eq:1} \\
&= |sw \dia F| + |wu| 
\quad \text{(since $wu$ has no fault)} \label{eq:2} \\
\intertext{Since $R$ intersects $xy$ above $e$ at a vertex $w$}
|st \dia F| 
&\le |sw \dia F| + |wt \dia F| \label{eq:3} \\
\intertext{As $wu$ is a subpath of $wt$,}
|wu| 
&< |wt| \le |wt \dia F| \label{eq:4} \\
\intertext{From \eqref{eq:2} and \eqref{eq:4}, we get}
|su \dia F| 
&< |sw \dia F| + |wt \dia F| \\
&\le |st \dia F|.
\end{align}

\end{proof}
\else
See the full version of the paper below for the proof of this lemma.
\fi
Using the above lemma, $v$ lies before $t$ in $L$. In the worst case, $v$ is the vertex preceding $t$ in the sequence. So, using the induction hypothesis, $\qu(v,\ro(\ft(sv)))$ returns a $(1+(k-1)^3\ep)$-approximation of $|sv \dia e|$. We are now ready to bound the answer returned by our algorithm. In 
\ifnum\myone=1
    \Cref{query:11}
\else
    Line \ref{query:11}
\fi
we set  
        \begin{align*}
        	\pa &= |\pr_\rr[v,t]| + \qu(v, \ro(\ft(sv)))\\
        	\intertext{Using \Cref{eq:one} and setting $Q=\pr_\rr[z,t]$, $S = xy$, we get:}
        	& = |Q| + |S[z,v]| +  \qu(v, \ro(\ft(sv)))
        \end{align*}

Thus $\pa$ is of the form described in \Cref{lem:belowe}. We now describe the approximation ratio of each component of this path.

\begin{enumerate}
\item $Q$ : After intersecting segment $xy$ at $z$, the path $R$ must continue along the shortest path $\pr_\rr =st$, so $|\pr_\rr[z,t]| = |R[z,t]|$. Thus, $Q=  \pr_\rr[z,t]$ gives a $1$-approximation of $|R[z,t]|$.

\item $S[z,v]$ : Now, $|S[z,v]| \le |S| =  |xy|$. By \Cref{lem:chechik}, $|xy| \leq \ep|\pr_\rr| \le \ep|R|$ (since $\pr_\rr$ is the shortest $st$ path, $|\pr_\rr| \le R$). Thus, $S = xy$ gives  a $\overbrace{\ep}^{\beta}$-approximation of $|R|$. 

\item $\qu(v, \ro(\ft(sv)))$: Using \Cref{lem:vless}, and the induction hypothesis,  $\qu(v, \ro(\ft(sv)))$ returns a $\overbrace{(1+(k-1)^3\ep)}^{\alpha}$ approximation of $|sv \dia e|$. 
\end{enumerate}

Substituting  $\alpha = 1+(k-1)^3\ep$, and $\beta = \ep$ into \Cref{lem:belowe}, we get $\pa \le (1 + k^3\ep)|R|$.

   This completes the correctness of our algorithm. We claim the following lemma:

   	\begin{lemma}
   	\label{lem:singlelemma}
   		There is a $(1+27\ep)$-approximate single fault tolerant distance oracle with  $\tl(n\sqrt n)$ space and $\tl(1)$ query time.
   	\end{lemma}

One could perform a more detailed analysis to achieve a better approximation ratio than in the above lemma. However, we prioritize conciseness, since this approach extends naturally to the dual-fault scenario. We now consider the more general case of two faulty edges.

\section{Extending our approach to 2 faults}

We extend our data structure and algorithm to handle dual faults by revisiting all components, including the $\ft$ structure and the query algorithm. While the overall framework remains similar to the single-fault case, key differences arise when handling two faults. In the single-edge fault setting, a detour of $st \dia e$ (see \Cref{detour_of_a_path}) starts on $st$.  However, for two edge faults, this definition is no longer sufficient. A path that avoids a faulty edge on the $st$ path may itself contain the second faulty edge. Therefore, the notion of a detour must be refined. We extend the definition of detours (from \Cref{detour_of_a_path}) to account for two edge faults.

 \begin{definition}[Detour of a path $R$ avoiding two edges]
 Let $P_1 = st$ and let $e_1$ be the last faulty edge on $P_1$.  Let $P_2$ be a path that avoids $e_1$ but contains the edge $e_2$. Let $R$ avoid $\{e_1,e_2\}$. Let $w$ be the last vertex of $R$ on $\pre(P_1,e_1) \cup \pre(P_2,e_2)$. Then $R[w,t]$ is the detour of $R$ and the detour starts at $w$. 
 \end{definition}
 
In the single fault case, we crucially used detours that start from a segment. If there is only one segment in question, then our \Cref{single_edge_detour} holds.  However, in a dual fault case, we may have to deal with two paths as alluded in the above definition. These two paths may have two different segments. Thus, the detour may start from the segment on the path $P_1$ (that contains the first fault $e_1$) or from the segment on the path $P_2$ (that contains the second fault $e_2$). Thus, we redefine the notion of a detour starting from a segment.

  \begin{definition}[Detour of a path $R$  starting from a segment]
  \label{def:starttwosegment}
  Let $P_1 = st$ and let $e_1$ be the last faulty edge on $P_1$.  Let $P_2$ be a path that avoids $e_1$ but contains the edge $e_2$. Let $R$ avoid $\{e_1,e_2\}$. Let $w$ be the last vertex of $R$ on $\pre(P_1,e_1) \cup \pre(P_2,e_2)$. We say that the detour of $R$ starts on the segment $seg(e_1,P_1)$ if 
  \begin{enumerate}
  	  	\item $w \in seg(e_1,P_1)$
  	  	\item $R[w,t]$ does not intersect with $seg(e_1,P_1) \cup seg(e_2,P_2)$  except at point $w$
  \end{enumerate} 
  Note that the second condition means that the detour $R[w,t]$ does not intersect either $seg(e_1,P_1)$ or $seg(e_2,P_2)$, except at $w$.
  
  Similarly, we say that $R$ starts on $seg(e_2, P_2)$ if in the first condition above $w \in seg(e_2,P_2)$. If $w$ lies both in $seg(e_1, P_1)$ and $seg(e_2, P_2)$, we give preference to the original shortest path $P_1$ and say that the detour starts on the segment $seg(e_1, P_1)$. See \Cref{second_type_of_detour}, \Cref{third_type_of_detour}, \Cref{fourth_type_of_detour} for illustrations.
 \end{definition}

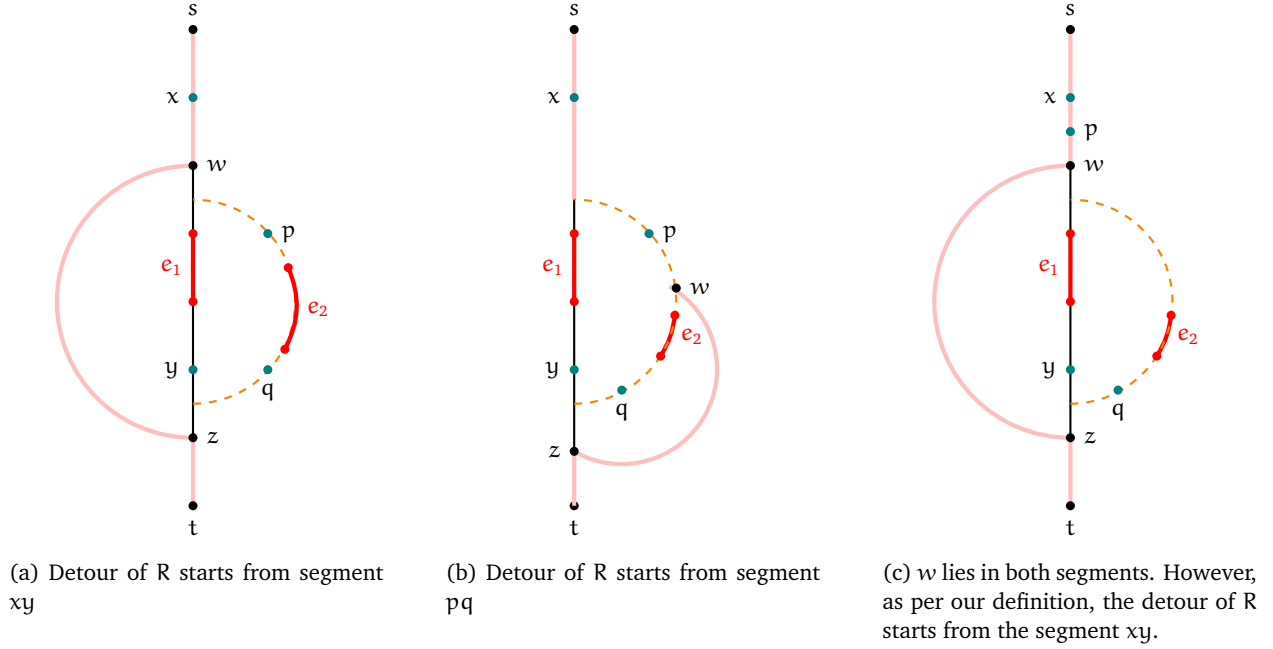
\begin{figure}
\centering
\setlength{\tabcolsep}{4pt}
\begin{minipage}[t]{0.30\textwidth}
    \centering
    \begin{tikzpicture}[scale=0.9, every node/.style={font=\small}]
        \tikzstyle{point} = [circle, fill=black, inner sep=1.2pt]
        \tikzstyle{mainpath} = [thick, black]
        \tikzstyle{detour} = [ultra thick, pink]
        \tikzstyle{edgefail} = [ultra thick, red]
        \tikzstyle{shortcut} = [orange, thick, dashed]

        \draw[mainpath] (0,7)--(0,0);
        \draw[edgefail] (0,4)--(0,3) node[midway,left]{$e_1$};
        \draw[detour] (0,5) arc[start angle=-90,end angle=90,radius=-2];
        \draw[detour] (0,1)--(0,0);
        \draw[detour] (0,5)--(0,7);
        \draw[shortcut] (0,1.5) arc[start angle=-90,end angle=90,radius=1.5];

        \node[point,label=above:$s$] at (0,7){};
        \node[point,label=left:$x$,teal] at (0,6){};
        \node[point,label=left:$y$,teal] at (0,2){};
        \node[point,red] at (0,4){};
        \node[point,red] at (0,3){};
        \node[point,label=right:$p$,teal] at (1.1,4){};
        \node[point,label=below:$q$,teal] at (1.1,2){};
        \node[point,label=right:$w$] at (0,5){};
        \node[point,label=below:$t$] at (0,0){};
        \node[point,label=right:$z$] at (0,1){};
        \node[point,red] at (1.4,3.5){};
        \node[point,red] at (1.35,2.3){};
        
        \draw[edgefail] (1.35,2.3) to[bend right=25] node[midway,right]{$e_2$} (1.4,3.5);
    \end{tikzpicture}
    \subcaption{Detour of $R$  starts from segment $xy$}
    \label{second_type_of_detour}
\end{minipage}
\hfill
\begin{minipage}[t]{0.30\textwidth}
    \centering
    \begin{tikzpicture}[scale=0.9, every node/.style={font=\small}]
        \tikzstyle{point} = [circle, fill=black, inner sep=1.2pt]
        \tikzstyle{mainpath} = [thick, black]
        \tikzstyle{detour} = [ultra thick, pink]
        \tikzstyle{edgefail} = [ultra thick, red]
        \tikzstyle{shortcut} = [orange, thick, dashed]

        \draw[mainpath] (0,7)--(0,0);
        \draw[edgefail] (0,4)--(0,3) node[midway,left]{$e_1$};
        \draw[edgefail] (1.27,2.2) to[bend right=12] node[midway,right]{$e_2$} (1.48,2.8);
        \draw[detour] (0,0.8) arc[start angle=-120,end angle=60,radius=1.4];
        \draw[detour] (0,4.5)--(0,7);
        \draw[shortcut] (0,1.5) arc[start angle=-90,end angle=90,radius=1.5];
        
        \node[point,red] at (0,4){};
        \node[point,red] at (0,3){};
        \node[point,red] at (1.27,2.2){};
        \node[point,red] at (1.48,2.8){};
        \node[point,label=above:$s$] at (0,7){};
        \node[point,label=left:$x$,teal] at (0,6){};
        \node[point,label=left:$y$,teal] at (0,2){};
        \node[point,label=right:$p$,teal] at (1.1,4){};
        \node[point,label=below:$q$,teal] at (0.7,1.7){};
        \node[point,label=right:$w$] at (1.5,3.2){};
        \node[point,label=below:$t$] at (0,0){};
        \draw[detour] (0,0)--(0,0.8);
        \node[point,label=left:$z$] at (0,0.8){};
    \end{tikzpicture}
    \subcaption{Detour of $R$ starts from segment $pq$}
    \label{third_type_of_detour}
\end{minipage}
\hfill
\begin{minipage}[t]{0.30\textwidth}
    \centering
    \begin{tikzpicture}[scale=0.9, every node/.style={font=\small}]
        \tikzstyle{point} = [circle, fill=black, inner sep=1.2pt]
        \tikzstyle{mainpath} = [thick, black]
        \tikzstyle{detour} = [ultra thick, pink]
        \tikzstyle{edgefail} = [ultra thick, red]
        \tikzstyle{shortcut} = [orange, thick, dashed]

        \draw[mainpath] (0,7)--(0,0);
        \draw[edgefail] (0,4)--(0,3) node[midway,left]{$e_1$};
        \draw[edgefail] (1.27,2.2) to[bend right=12] node[midway,right]{$e_2$} (1.48,2.8);
        \draw[detour] (0,5)--(0,7);
        \draw[detour] (0,1)--(0,0);
        \draw[detour] (0,1) arc[start angle=90,end angle=-90,radius=-2];
        \draw[shortcut] (0,1.5) arc[start angle=-90,end angle=90,radius=1.5];

        \node[point,label=above:$s$] at (0,7){};
        \node[point,red] at (0,4){};
        \node[point,red] at (0,3){};
        \node[point,red] at (1.27,2.2){};
        \node[point,red] at (1.48,2.8){};
        \node[point,label=left:$x$,teal] at (0,6){};
        \node[point,label=left:$y$,teal] at (0,2){};
        \node[point,label=right:$w$] at (0,5){};
        \node[point,label=below:$q$,teal] at (0.7,1.7){};
        \node[point,label=right:$p$, teal] at (0,5.5){};
        \node[point,label=below:$t$] at (0,0){};
        \node[point,label=right:$z$] at (0,1){};
    \end{tikzpicture}
    \subcaption{$w$ lies in both segments. However, as per our definition, the detour of $R$ starts from the segment $xy$.}
    \label{fourth_type_of_detour}
\end{minipage}

\caption{Illustration of different types of detours in path $R$ under one or two edge failures. The straight line from top to bottom represents the $st$ path (or path $P_1$), while the dashed path containing edge $e_2$ represents path $P_2$. $e_1$ lies on segment $xy$ of $P_1$ and $e_2$ lies on segment $pq$ of $P_2$. $R$ is represented by the path coloured pink.
}
\label{fig:basic_two_edge_fault}
\end{figure}

\ifnum\myone=1
We now describe the construction of $\ft(st)$ using the function $\mn$ defined in \Cref{alg:constructftgeneral}. This function is quite similar to \Cref{alg:constructft} -- we have highlighted major changes in gray. This function takes one more extra parameter compared to  \Cref{alg:constructft}. We will describe its use later. \Cref{alg:constructft} constructs a tree that has two levels. But \Cref{alg:constructftgeneral} will construct a tree that has three levels -- level 0, level 1 and level 2, where level 2 nodes are the leaves of the tree. We will now describe the construction in detail, starting with level 0 vertices. We make the root of $\ft(st)$ by invoking  $\mn(G,st,\emptyset, 0)$. At the root $\rr$, the primary path is $\pr_\rr = st$. At the root, we make a segment and detour children at level 1, which we describe next.

\begin{algorithm}[ht!]
Make a node $\rr$ with $\pr_\rr = P$ and $\bu(\rr) = B$\\
\If{$depth  < 2$}
{
	\If{ $\bu(\rr)$ is empty}
	{
		\ForEach{segment $xy$ in $\pr_\rr$}
		{
        
            Let $P(xy)$ be the shortest 2-decomposable path $Q_1 \odot e_1 \odot Q_2 \odot e_2 \odot Q_3$
            in the original graph $G$, where each $Q_j$ is a shortest path in $G$, such that $P(xy)$ survives in $H-xy$
            
            \If{$P(xy)$ exists}
            {
            a child of $\rr$ $\leftarrow \mn(H-xy,P(xy), \emptyset, depth+1)$\\
            }\label{line:line7}
            \Comment{\textcolor{blue}{ 1-decomposable detour children}}
            \ForEach{$i = 0$ to $\log_\Ep nW$}
                {
                Let $\mathcal{Q}= \{ P \dia e_1 ~|~ e_1 \in xy~\}$
                
                Let $P_i$ be the path in $\mathcal{Q}$ whose detour starts closest to $x$ on the segment $xy$ and its detour joins below $y$ on $\pr_\rr$ and  the detour length $\le (\Ep)^i$ 
                
            \If{$P_i$ exists (i.e., $\mathcal{Q} \neq \emptyset$)}
            {
                a child of $\rr \leftarrow \mn(H,P_i, \{xy\}, depth+1)$
                }\label{line:line11} 
             }
             \Comment{\textcolor{blue}{ 2-decomposable detour children}}
             \colorbox{gray!20}{%
	 \parbox{\dimexpr\linewidth-18\fboxsep}{
             \ForEach{$i = 0$ to $\log_\Ep nW$}
                {
                Let $\mathcal{L}= \{ st \dia \{e_1,e_2\} ~|~ e_1 \in xy~\ and ~e_2 \in (H-xy)~\}$ \label{line:2-decompose}
                
                Let $P_i$ be the path in $\mathcal{L}$ whose detour starts closest to $x$ on the segment $xy$ and its detour joins below $y$ on $\pr_\rr$ and  the detour length $\le (\Ep)^i$ 
                
            \If{$P_i$ exists (i.e., $\mathcal{L} \neq \emptyset$)}
            {
                a child of $\rr \leftarrow \mn(H,P_i, \{xy\}, depth+1)$
                } 
             }
	    
         }\label{line:line13}
	    }
    }
    }
    
	\Else
	{
	\colorbox{gray!20}{%
	 \parbox{\dimexpr\linewidth-15\fboxsep}{%
		Let $\bu(\rr) = \{xy\}$\\
		\ForEach{segment $pq$ in $\pr_\rr$}
		{
			\ForEach{$i = 0$ to $\log_\Ep nW$}
			{

			\Comment{\textcolor{blue}{2-decomposable detour children}}
			Let $\mathcal{Q}= \{ P \dia F ~|~ \text{$F = \{e,e'\}$ where  $e \in xy$ and $e' \in pq$} \}$

            Let $P_i$ be the path in $\mathcal{Q}$ whose detour starts closest to $x$ on the segment $xy$ and its detour joins below $q$ on $\pr_\rr$ and  the detour length $\le (\Ep)^i$ 
            
            \If{$P_i$ exists}
            {
            	a child of $\rr \leftarrow \mn(H,P_i, \{xy\}, depth+1)$
            }\label{line:line15}
            
            Let $S_i$ be the path in $\mathcal{Q}$ whose detour starts closest to $p$ on the segment $pq$ and its detour joins below $q$ on $\pr_\rr$ and  the detour length $\le (\Ep)^i$  
                
            \If{$S_i$ exists}
            {  
            	a child of $\rr\leftarrow \mn(H, S_i, \{pq\}, depth+1)$
            }	
			} 
		}
	}
	}}

} 	
return $\rr$
\caption{$\mn(H, P, B, depth)$}	
\label{alg:constructftgeneral}
\end{algorithm}
\
 
  \subsection{Level 1}

The level 1 nodes constructed in our algorithm are nearly similar to  \Cref{alg:constructft}. The difference lies in the following step. 
 \begin{enumerate}
 
 \item Segment Child : For a segment $xy$ in $\pr_\rr$, we find the shortest {\bf 2-decomposable} path in the original graph $G$ that survives in $G-xy$.
 
 \item Detour Child: Here we have two types of detour children. 
 
    \begin{enumerate}
     \item 1-decomposable detour child

     One is of the similar type as in \Cref{alg:constructft} where we have stored a detour path $P_i$ whose detour starts as close to as $x$ on the segment $xy$ and joins below the segment $xy$ and its detour length is $\leq (1+\epsilon)^i$. Note that $P_i$, if it exists,  is a 1-decomposable path. If such detour $P_i$ exist,  then we create a child $\rr_i$ of $\rr$ by calling $\mn(G, P_i, \{xy\}, 1)$, set its primary path to $P_i$. We call this detour a 1-decomposable detour as the path $P_i$ is 1-decomposable. Note that in the single fault tolerant distance oracle, we only constructed 1-decomposable detour children and that is why we never gave them an explicit name. But for dual faults, we will construct 1-decomposable as well as 2-decomposable detour children (which we will describe in the next enumeration).
 
    Unlike segment children, $\rr_i$ does not arise by removing a segment from $\pr_\rr$. Specifically, the detour of $\rri$ starts from a segment of $\pr_\rr$. To mark this exception, we set $\bu(\rr_i) = \{xy\}$, indicating that its primary path includes a detour originating in the segment $xy$. 

    \item 2-decomposable detour child

     There will be another kind of detour child which is highlighted by gray in \Cref{alg:constructftgeneral}.  $\mathcal{L}$ contains a set of shortest paths avoiding two edges with one edge lying on the segment $xy$ and another edge can be anywhere in the graph. Let $P_i$ be  path in $\mathcal{L}$  whose detour starts as close to $x$ on the segment $xy$, joins below $y$ on $\pr_\rr$, with detour length at most $(\Ep)^i$. Note that $P_i$, if it exists, is 2-decomposable. We create a child $\rr_i$ of $\rr$ by calling $\mn(G, P_i, \{xy\}, 1)$, set its primary path to $P_i$, and define $\bu(\rr_i)=\{xy\}$.
    \end{enumerate}

 As in the single fault case, there are $O(\loge)$ segment  and $O(\logee)$ detour children at level 1.
 \end{enumerate}

\subsection{Level 2}

For a \emph{segment} level-1 node, we recursively apply the same process used at the root. Since there are at most $O(\loge)$ segment nodes at level 1, each contributing up to $O(\logee)$ children at level 2. Therefore, the total number of level-2 nodes added while processing segment level-1 nodes is $O(\logt)$.

Consider a node $\rr$ at level~1 with $\bu(\rr) = \{xy\}$. By construction, $\pr_\rr$ is a path surviving the failure of an edge in $xy$. To handle a second fault, we iterate through every segment $pq$ on the current path $\pr_\rr$. For each segment $pq$, we consider the set of paths $\mathcal{Q}$ that survive the failure of two edges $\{e_1, e_2\}$ (where $e_1 \in xy, e_2 \in pq$). From this set, we extract two 2-decomposable detours for the next level. Note that these detours are defined with respect to the two segments $xy$ and $pq$, and thus follow the detour starting rules outlined in Definition \Cref{def:starttwosegment}

\begin{enumerate}
    
    \item 
Let $P_i \in \mathcal{Q}$ be the path whose detour starts \textbf{closest to $x$} on the segment $xy$ and its detour joins below $q$ on $\pr_\rr$ and  the detour length $\le (\Ep)^i$ . This path handles detour that start from the segment $xy$. We create a child $\rr_i$ of $\rr$ by invoking $\mn(G, P_i, \{xy\}, 2)$.

    \item Let $S_i \in \mathcal{Q}$ be the path whose detour starts \textbf{closest to $p$} on the segment $pq$ and its detour joins below $q$ on $\pr_\rr$ and  the detour length $\le (\Ep)^i$ . This path handles the  detour that starts in the segment $pq$. We create a child $\rr_i$ of $\rr$ by invoking $\mn(G, S_i, \{pq\}, 2)$. Here, we update the set to $\bu(\rr_i) = \{pq\}$ to explicitly record that the path depends on the detour originating from the segment $pq$.

\end{enumerate}

For each segment $pq \in \pr_\rr$, we create at most $O(\loge)$ children at level 2 of $\ft(st)$ (from both cases (1) and (2)). Since $\pr_\rr$ contains $\loge$ segments, each detour node $\rr$ contributes $O(\logee)$ children. With $O(\logee)$ detour nodes coming from both 1-decomposable and 2-decomposable children at level 1, the total number of level-2 nodes arising from detour nodes is $O(\logf)$. We thus claim the following:

\begin{lemma}
	There are $O(\logf)$ nodes in $\ft(st)$.
\end{lemma}

As in the single fault case, we claim that the path in each node of $\ft(st)$ is 2-decomposable and we can represent each node of $\ft(st)$ in $\tl(\sqrt n)$ space. Thus, we claim the following lemma.

\begin{lemma}
    For each $t$ ,  $\ft(st)$ can be represented using $\tl(\sqrt n)$ space. Thus, the total space taken by $\ft(st)$ overall $t \in V \setminus \{s\}$ is $\tl(n\sqrt n)$. Moreover, the time taken to find if $e \in \pr_\rr$ is $\tl(1)$ and we can also report the segment of $\pr_\rr$ that contains $e$ in $\tl(1)$ time.
\end{lemma}
\else
	In the full version of the paper below, we show how to extend our approach for two faults. Due to lack of space, we refer the reader to the full version of the paper for the construction of $\ft$ data-structure, the query algorithm, running time as well as correctness.
\fi

\section{Query Algorithm}
	\begin{algorithm}[t]
	\DontPrintSemicolon
	\colorbox{gray!20}{%
	 		\parbox{\dimexpr\linewidth-15\fboxsep}{%
	 \If{any ancestor function call of $\quDtwo$ had $t$ as the first parameter and $\rr$ is not a segment child}
	{
		return $\infty$
	}
	}}
	
	\If{$\pr_\rr$ does not contain any faulty edge}{
		\Return $|\pr_\rr|$
	}
	\Else{
		$\pa \leftarrow \infty$
		
		Let $e = (u,v)$ be the last faulty edge on $\pr_\rr$, and let $xy \in \pr_\rr$ be the segment containing $e$\\
		\Comment{\textcolor{blue}{ Part 1: $st \dia F$ avoids segment $xy$ }}
		Let $\xx$ be the child of $\rr$ that does not contain $xy$\\
		$\pa \leftarrow \min\{\pa,\quDtwo(t,\xx)\}$
		\label{querynormal:8}

		\Comment{\textcolor{blue}{ Part 2: $st \dia F$ intersects segment $xy$ below $e$ }}
		\If{$v \neq t$}
		{
		
		$\pa \leftarrow \min\{\pa,  \quDtwo(v, \ro(\ft(sv))) + |\pr_\rr[v,t]|\}$
		\label{querynormal:10}
		}
		
		\Comment{\textcolor{blue}{ Part 3: $st \dia F$ intersect segment $xy$ above $e$}}
        \colorbox{gray!20}{%
	 		\parbox{\dimexpr\linewidth-15\fboxsep}{%
        $j \leftarrow 1$}}
        \While{ $j \le 2$}
        {
            \ForEach{$i = 0$ to $\log_\Ep n$ }{
			\colorbox{gray!20}{%
	 		\parbox{\dimexpr\linewidth-15\fboxsep}{Let $\rr_i$ be the $j$-decomposable detour child of $\rr$ such that the  detour of $\pr_{\rr_i}$ starts on $xy$ and joins below $xy$ and has length $ \leq (\Ep)^i$ \label{line:12}}}
			
			\colorbox{gray!20}{%
	 		\parbox{\dimexpr\linewidth-15\fboxsep}{\If{$\rri$ exists and $\pr_{\rr_i}$ avoids $F$}
			{
				    $\pa \leftarrow \min\{\pa, |\pr_{\rr_i}|\}$%
                    \label{qn:15}%
			}}}
			
			\colorbox{gray!20}{%
	 		\parbox{\dimexpr\linewidth-15\fboxsep}{\ElseIf{$\pr_{\rr_i}$ contains a faulty edge of $F$} 
			{   \label{line:15}
				$\pa \leftarrow \min\{\pa, \quDone(t, \rr_i)\}$
			}}}
            }
            \colorbox{gray!20}{%
	 		\parbox{\dimexpr\linewidth-15\fboxsep}{
            $j \leftarrow j+1$}}
			}
		}
	\Return $\pa$
	\caption{$\quDtwo(t, \rr$)}
	\label{querynormal}
\end{algorithm}

Given a destination vertex $t$ and a fault set $F$, we execute the \Cref{querynormal}, $\quDtwo(t, \ro(\ft(st)))$, assuming that $F$ is implicitly known. We now compare $\quDtwo$ with our single fault $\qu$ algorithm. The two algorithms are almost identical, except for a few modifications, which we have highlighted in gray in the pseudocode of $\quDtwo$.

We examine the first {\em if} condition in $\quDtwo$. Consider the recursion tree of our query algorithm with the root representing $\quDtwo(t,\ro(\ft(st)))$. Note that  $\quDtwo$ can recursively call itself and may also call  $\quDone$. Thus, the recursion tree may contain both $\quDtwo$ and $\quDone$. However,  $\quDone$ makes calls to $\quDtwo$ only. Thus, each call of $\quDone$ is necessarily followed by $\quDtwo$ in the recursion tree. Thus, we can just focus on the number of calls to $\quDtwo$ on any path in the recursion tree.

Unlike $\qu$, the procedure $\quDtwo$ may be invoked again with the same parameter, say $t$, again. 
\ifnum\myone=1
    \Cref{querynormal:8}
\else
    Line \ref{querynormal:8}
\fi
 triggers such a call. During this recursion, the algorithm only moves to a segment child in the $\ft(st)$ tree. After at most two such recursive calls, the execution reaches a leaf of $\ft(st)$. At that point, the algorithm cannot invoke $\quDtwo$ again with parameter $t$, and every subsequent recursive call must use a new affected vertex. We maintain the invariant that once $\quDtwo$ moves to a new vertex in a path in the recursive tree, it never invokes $\quDtwo$ again with the same parameter $t$. The first {\em if} condition in $\quDtwo$ enforces this invariant.

 The other change is as follows: In Part 3 of $\quDtwo$, we find both 1-decomposable and 2-decomposable detour children of $\rr$. Let us call this node $\rri$. If the primary path in $\rr_i$ avoids $F$, then we process it in a way similar to what was done in $\qu$. However, since there are two faults, it may be the case that $\pr_{\rr_i}$ contains the second fault of $F$. This is the {\bf if} condition in 
 \ifnum\myone=1
    \Cref{line:15}.
\else
    Line \ref{line:15}.
\fi If this condition is satisfied, then we process the node $\rr_i$ in $\ft(st)$ using the function $\quDone$. This completes the description of $\quDtwo$. We now describe $\quDone$.

\subsection{\texorpdfstring{Our algorithm $\quDone$}{Our algorithm quDone}}	
	\begin{algorithm}[t]
       	Let us assume that $\pr_{\rr_i}$ contains $e' = (u',v')$ in a segment $pq$

			\Comment{\textcolor{blue}{$st \dia F$ does not intersect $pq$ and or the detour of $\stf$ starts on the segment $xy$}}
            $\fpa \leftarrow \infty$
            
			\ForEach{$j = 0$ to $\log_\Ep n$}{
            
			Let $\xx_j$ be the detour child of $\rr_i$ such that the detour of $\pr_{\xx_j}$ starts on segment $xy$ at, say vertex $w$, and joins below the segment $pq$ and the detour length is  $\le (\Ep)^j$\\
			Let $r$ be the vertex nearest to $w$ in the segment $xy$ where $r  \in \{u,u'\}$\\
			$\fpa \leftarrow \min\{ \fpa, |\pr_{\xx_j}[w,t]| + |\pr_{\rr_i}[w,r| + \quDtwo(r,\ro(\ft(sr)))\}$\\
			\label{querydetourone:5}

			}
			\Comment{\textcolor{blue}{$\stf$ intersects $pq$ below $e'$}}
			\If{$v' \neq t$}
			{
			$\spa \leftarrow \quDtwo(v', \ro(\ft(sv'))) + |\pr_{\rr_i}[v',t]|$\\
			\label{querydetour:6}
			}
		\Comment{\textcolor{blue}{The detour of $\stf$ starts on $pq$ above $e'$}}
        $\tpa \leftarrow \infty$
		\ForEach{$j = 0$ to $\log_\Ep n$}{
			
			Let $\yy_j$ be the detour child of $\rr_i$ such that the detour of $\pr_{\yy_j}$ starts on segment $pq$, say at vertex $w'$ and joins below $q$ with detour length $\le (\Ep)^j$\\
			$\tpa \leftarrow \min\{ \tpa, \quDtwo(u', \ro(\ft(su'))) + |\pr_{\rr_i}[u',w']| + |\pr_{\yy_j}[w',t]|\}$\\
			\label{querydetour:9}

			}
		\Return $\min\{\fpa,\spa,\tpa\}$
	\caption{\quDone($t, \rr_i$)}
	\label{querydetourone}
\end{algorithm}

	Let us now look at the case when we call $\quDone(t,\rr_i)$. To this end, let us first set up some notations. Let $\rr$ be the root of $\ft(st)$ and $e$ be the last edge of $F$ on $\pr_\rr$. Let $e$ lie on segment $xy$ on $\pr_\rr$.  $\rr_i$ is a detour node and the detour of $\pr_{\rr_i}$ starts above $e$ on the segment $xy$ on the $st$ path. We call $\quDone(t,\rr_i)$ only if $\pr_{\rr_i}$ contains the second fault edge from $F$. So, let us assume that $e' = (u',v')$ is this faulty edge on $\pr_{\rr_i}$. Let $pq$ be the segment of $\pr_{\rr_i}$ on which $e'$ lies. As in $\quDtwo$, we now check how $\stf$ interacts with segment $pq$. There are three cases, and we go over them.

	\begin{enumerate}
		\item $\stf$ does not intersect with $pq$ or its detour starts in the segment $xy$. (See \Cref{Not_intersecting_with_pq} and \Cref{Intersecting_after_q})

		While processing $\rr_i$, we constructed detour children whose primary path's detour starts from the segment $xy$. We will later show that one of them, say $\xx_j$, yields a good approximation of $|\stf|$. Specifically, in 
         \ifnum\myone=1
            \Cref{querydetourone:5}
        \else
            Line \ref{querydetourone:5}
        \fi
        , we identify the vertex $w$ on the segment $xy$ using $\pr_{\xx_j}$, locate the nearest affected vertex $r$ (where $r$ may be either $u$ or $u'$) on $xy$, and then invoke $\quDtwo(r, \ro(\ft(sr)))$. We will subsequently prove that the path computed in 
        \ifnum\myone=1
            \Cref{querydetourone:5}
        \else
            Line \ref{querydetourone:5}
        \fi provides a good approximation of $|\stf|$.

	\item $\stf$ intersects $pq$ below $e'$. (See \Cref{Intersecting_below_e_2})
		
	This case parallels Part~2 of the procedure $\quDtwo$. The only difference is that $\quDtwo$ locates the affected vertex $v$ on the shortest $st$ path, whereas here we locate the affected vertex $v'$ in 
    \ifnum\myone=1
            \Cref{querydetour:6}
        \else
            Line \ref{querydetour:6}
        \fi along the primary path $\pr_{\rr_i}$.

    \begin{figure}
    \centering
    \begin{minipage}[t]{0.25\textwidth}
        \centering
        \begin{tikzpicture}[scale=0.8, every node/.style={font=\small}]
        \tikzstyle{point} = [circle, fill=black, inner sep=1.5pt]

        \draw[thick] (0,7) -- (0,0);

        \draw[red, ultra thick] (0,3.5) -- (0,2.5) node[midway, left] {\scriptsize $e$};
        \draw[magenta, line width=0.8mm, opacity=0.7] (0,4.2) arc[start angle=-90, end angle=90, radius=-1.5]node[pos=0.5, below=14, sloped, black, rotate=180] {\scalebox{0.7}{$\stf$}};

        \draw[pink, line width=0.8mm, opacity=0.7] (0,1.2) -- (0,0);
        \draw[pink, line width=0.8mm, opacity=0.7] (0,4.2) -- (0,7);
        \draw[orange, thick] (0,1.5) arc[start angle=-90, end angle=90, radius=1.5];
        \draw[red, ultra thick] (1.4,2.4) to[bend right=18] node[midway, right] {\scriptsize $e'$} (1.4,3.5);
        
        \node[point, label=below:{\scalebox{0.7}{$q$}} ,teal] at (1.1,2) {};
        \node[point, magenta,label=right:{\scalebox{0.7}{$w$}} ] at (0,4.2) {};
        \node[point, red, label=right:{\scalebox{0.7}{$v$}} ] at (0,2.5) {};
        \node[point, red, label=right:{\scalebox{0.7}{$u$}} ] at (0,3.5) {};
        \node[point, red, label=right:{\scalebox{0.7}{$v'$}} ] at (1.4,2.4) {};
        \node[point, red, label=right:{\scalebox{0.7}{$u'$}} ] at (1.4,3.5) {};
        \node[point, label=right:{\scalebox{0.7}{$p$}}, teal] at (1.1,4) {};
        \node[point, teal, label=left:{\scalebox{0.7}{$x$}}] at (0,6) {};
        \node[point, teal, label=left:{\scalebox{0.7}{$y$}}] at (0,2) {};
        \node[point, label=above:$s$] at (0,7) {};
        \node[point, label=below:$t$] at (0,0) {};
        \end{tikzpicture}
        \subcaption{\scalebox{0.9}{$\stf$ does not intersect $pq$}}
        \label{Not_intersecting_with_pq}
    
    \end{minipage}
    \hfill
    \begin{minipage}[t]{0.25\textwidth}
        \centering
        \begin{tikzpicture}[scale=0.8, every node/.style={font=\small}]
        \tikzstyle{point} = [circle, fill=black, inner sep=1.5pt]
        \tikzstyle{segment} = [thick]

        \draw[thick] (0,7) -- (0,0);

        \draw[segment] (0,6) -- (0,2);
        
        \draw[red, ultra thick] (0,4) -- (0,3) node[midway, left] {\scriptsize $e$};
        \draw[orange, thick] (0,1.5) arc[start angle=-90, end angle=90, radius=1.75];

        \draw[magenta, line width=0.8mm, opacity=0.7]
        (0.7,1.6) .. controls (1.2,3) and (0.5,4.2) .. (0,4.7)node[pos=0.5, below=-2, sloped, black, rotate=360] {\scalebox{0.7}{$\stf$}};

        \draw[pink, line width=0.8mm, opacity=0.7] (0,1.5) -- (0,0);

        \draw[pink, line width=0.8mm, opacity=0.7] (0,4.5) -- (0,7);

        \node[point, red, label=left:{\scalebox{0.7}{$v$}} ] at (0,3) {};
        \node[point, red, label=left:{\scalebox{0.7}{$u$}} ] at (0,4) {};
        \node[point, magenta,label=left:{\scalebox{0.7}{$w$}} ] at (0,4.7) {};
        \node[point, label=right:{\scalebox{0.7}{$p$}}, teal] at (1.1,4.6) {};
        \node[point, teal, label=left:{\scalebox{0.7}{$x$}}] at (0,6) {};
        \node[point, teal, label=left:{\scalebox{0.7}{$y$}}] at (0,2) {};
        \node[point, label=below:{\scalebox{0.7}{$q$}},teal] at (1.1,1.9) {};
        \draw[red, ultra thick] (1.6,2.5) to[bend right=20] node[midway, right] {\scriptsize $e'$} (1.6,4);
        \node[point, red, label=right:{\scalebox{0.7}{$u'$}} ] at (1.6,4) {};
        \node[point, red, label=right:{\scalebox{0.7}{$v'$}} ] at (1.6,2.5) {};
         \node[point, label=above:$s$] at (0,7) {};
        \node[point, label=below:$t$] at (0,0) {};
        \end{tikzpicture}
        \subcaption{\scalebox{0.9}{$\stf$ does not intersect $pq$}}
        \label{Intersecting_after_q}
    \end{minipage}
    \hfill
        \hfill
    \begin{minipage}[t]{0.24\textwidth}
        \centering
        \begin{tikzpicture}[scale=0.8, every node/.style={font=\small}]
        \tikzstyle{point} = [circle, fill=black, inner sep=1.5pt]
        \tikzstyle{segment} = [thick]

        \draw[thick] (0,7) -- (0,0);

        \draw[segment] (0,6) -- (0,2);
        
        \draw[red, ultra thick] (0,4) -- (0,3) node[midway, left] {\scriptsize $e$};
        \draw[orange, thick] (0,1.5) arc[start angle=-90, end angle=90, radius=1.75];

        \draw[magenta, line width=0.8mm, opacity=0.7]
        (1.4,2.2) .. controls (1.2,3) and (0.5,4.2) .. (0,4.5)node[pos=0.5, below=-2, sloped, black, rotate=360] {\scalebox{0.7}{$\stf$}};

        \draw[pink, line width=0.8mm, opacity=0.7] (0,1.5) -- (0,0);

        \draw[pink, line width=0.8mm, opacity=0.7] (0,4.5) -- (0,7);

        \node[point, label=right:{\scalebox{0.7}{$p$}}, teal] at (1.1,4.6) {};
        \node[point, teal, label=left:{\scalebox{0.7}{$x$}}] at (0,6) {};
        \node[point, teal, label=left:{\scalebox{0.7}{$y$}}] at (0,2) {};
        \node[point, label=below:{\scalebox{0.7}{$q$}},teal] at (1.1,1.9) {};
        \draw[red, ultra thick] (1.6,2.5) to[bend right=20] node[midway, right] {\scriptsize $e'$} (1.6,4);
        \node[point, red, label=left:{\scalebox{0.7}{$v$}} ] at (0,3) {};
        \node[point, magenta,label=left:{\scalebox{0.7}{$w$}} ] at (0,4.5) {};
        \node[point, red, label=left:{\scalebox{0.7}{$u$}} ] at (0,4) {};
        \node[point, red, label=right:{\scalebox{0.7}{$u'$}} ] at (1.6,4) {};
        \node[point, red, label=right:{\scalebox{0.7}{$v'$}} ] at (1.6,2.5) {};
        \node[point, label=above:$s$] at (0,7) {};
        \node[point, label=below:$t$] at (0,0) {};
        \end{tikzpicture}
        \subcaption{\scalebox{0.9}{$\stf$ intersects $pq$ below $e^{'}$}}
        \label{Intersecting_below_e_2}
    \end{minipage}
    \hfill
    \begin{minipage}[t]{0.24\textwidth}
        \centering
        \begin{tikzpicture}[scale=0.8, every node/.style={font=\small}]
        \tikzstyle{point} = [circle, fill=black, inner sep=1.5pt]
        \tikzstyle{segment} = [thick]

        \draw[thick] (0,7) -- (0,0);

        \draw[segment] (0,6) -- (0,2);
        
        \draw[red, ultra thick] (0,4) -- (0,3) node[midway, left] {\scriptsize $e$};
        \draw[orange, thick] (0,1.5) arc[start angle=-90, end angle=90, radius=1.5];
        \draw[magenta, line width=0.8mm, opacity=0.7] (0,0.8) arc[start angle=-120, end angle=68, radius=1.55]node[pos=0.5, below=-2, sloped, black, rotate=360] {\scalebox{0.7}{$\stf$}};
        \draw[pink, line width=0.8mm, opacity=0.7] (0,0.8) -- (0,0);
        \draw[pink, line width=0.8mm, opacity=0.7] (0,4.5) -- (0,7);

        \node[point, label=below:{\scalebox{0.7}{$q$}},teal] at (0.7,1.7) {};
        \node[point, label=right:{\scalebox{0.7}{$p$}}, teal] at (1.1,4) {};
        \node[point, teal, label=left:{\scalebox{0.7}{$x$}}] at (0,6) {};
        \node[point, teal, label=left:{\scalebox{0.7}{$y$}}] at (0,2) {};
        \draw[red, ultra thick] (1.27,2.2) to[bend right=12] node[midway, right] {\scriptsize $e'$} (1.48,2.8);

        \node[point, magenta,label=left:{\scalebox{0.7}{$w'$}} ] at (1.41,3.55) {};
        \node[point, red, label=left:{\scalebox{0.7}{$v$}} ] at (0,3) {};
        \node[point, red, label=left:{\scalebox{0.7}{$u$}} ] at (0,4) {};
        \node[point, red, label=left:{\scalebox{0.7}{$v'$}} ] at (1.27,2.2) {};
        \node[point, red, label=left:{\scalebox{0.7}{$u'$}} ] at (1.48,2.8) {};
        \node[point, label=above:$s$] at (0,7) {};
        \node[point, label=below:$t$] at (0,0) {};
        \end{tikzpicture}
        \subcaption{detour of $\stf$ starts from the segment $pq$}
        \label{Intersecting_below_p}
    \end{minipage}
    \hfill
    \caption{Different ways $\stf$ can intersect.}
    \label{fig:overall_two_edge_fault}
\end{figure}
    
	\item  The detour of $\stf$ starts from $pq$ above $e'$. (See \Cref{Intersecting_below_p})
		
	To handle this case, we have precomputed several 2-decomposable detour children of $\rr_i$ whose detours start on the segment $pq$. We will later show that one of them, say $\yy_j$, closely approximates the detour of $\stf$. Let the detour of $\pr_{\yy_j}$ start at $w'$. As in Case~(1), in 
    \ifnum\myone=1
        \Cref{querydetour:9},
    \else
        Line \ref{querydetour:9},
    \fi  we move from $w'$ to the nearest affected vertex $u'$ on $\pr_{\rr_i}$ and recursively invoke $\quDtwo(u',\ro(\ft(su')))$. We will prove later that the path calculated in 
    \ifnum\myone=1
        \Cref{querydetour:9}
    \else
        Line \ref{querydetour:9}
    \fi provides a good approximation of $|\stf|$.

	\end{enumerate}

This completes the description of $\quDone$.

\subsection{Running Time}

We first analyze the running time of our algorithm, which closely follows the analysis in \Cref{sec:onefaultrunning}. Consider the recursion tree where each node corresponds to a call to $\quDtwo$ or $\quDone$, with the root representing $\quDtwo(t,\ro(\ft(st)))$. Each call to $\quDtwo$ or $\quDone$ performs $O(\loge)$ non-recursive work and makes at most $O(\loge)$ recursive calls. We will show that the height of this tree is at most~40, giving a total of $O(\log_{1+\ep}^{40} nW)$ nodes in the recursion tree. Hence, the overall running time of the algorithm is $O(\log_{1+\ep}^{41} nW)$.

Each node in this tree corresponds to either $\quDtwo$ or $\quDone$. In $\quDone$, we make recursive calls to $\quDtwo$, ensuring that at least half of the nodes on any path from the root are $\quDtwo$ calls. So, let us find how many times we can call $\quDtwo$ on any path in this recursion tree.

 Let $v \in V(F) \cup \{t\}$. We want to find how many times $\quDtwo(v,\cdot)$ will be called. We claim that this can be called at most 4 times. Indeed as explained during the discussion of $\quDtwo$, every time $\quDtwo(v,\cdot)$ is called, we are moving to a segment child of  $\ft(sv)$. Since the height of $\ft(sv)$ is 3, the total number of calls to $\quDtwo$ with $v$ as a first parameter can be at most  3. After these 3 calls, any subsequent call to $\quDtwo$ must use a new vertex from $V(F) \cup \{t\}$. The first {\em If} condition in $\quDtwo$ ensures that if we call $\quDtwo$ with the parameter $v$ again in the future, then our algorithm returns $\infty$ and recursion stops. Thus, on any recursion path there are atmost 4 calls to $\quDtwo$ with parameter $v$.   Since there are at most 5 vertices in $V(F) \cup \{t\}$, the total number of recursive calls of $\quDtwo$ is at most 20.  Accounting for interspersed $\quDone$ calls, the tree height is at most~40, giving a total of $O(\log_{1+\ep}^{40} nW)$ nodes. We thus obtain the following lemma.

\begin{lemma}
	The running time of $\quDtwo$ is $O(\log_{1+\ep}^{41} nW)$.
\end{lemma}

\section{Correctness} 
As in the single-fault case, our algorithm tackles two main challenges: (1) when $\stf$ intersects a segment containing a faulty edge below $e$, and (2) when the detour of $\stf$ starts above $e$ on the same segment. With two faults, two such segments exist. For challenge~(1), we directly reuse \Cref{lem:belowe}. The lemma is sufficiently general and applies unchanged to the dual fault case. We handle challenge~(2) similarly but omit a separate lemma, addressing it directly when the case arises. 

In the function $\quDtwo$, the first parameter is either $t$ or an affected vertex. Let us order the vertices in $V(F) \cup \{t\}$ by their distance from $s$, avoiding the fault set $F$. Let $L$ be the subsequence of all vertices up to $t$ sorted in this order. The size of $L$ is at most $5$. We now prove the following lemma. 

\begin{lemma} \label{lem:kk} 
Let $p$ be the $k$-th vertex in the sequence $L$ where $k \le 5$. Then $\quDtwo(p, \ro(\ft(sp)))$ returns a $(1 + k^5\ep)$-approximation of $|sp \dia F|$. 
\end{lemma} 

We prove the lemma by induction on the vertices of $L$. For the base case ($k=1$), without loss of generality, let $u$ be the first vertex in the sequence. The shortest $su$ path avoids the faulty edge $F$, and $\quDtwo(u, \ro(\ft(su)))$ returns $|su|$ exactly, establishing the base case. 

For the induction step, let us assume the lemma holds for the $k-1$ vertices before $p$ in $L$. We now prove it for the $k$-th vertex $p$. For notational convenience, we relabel $p$ as $t$, as this notation aligns with all our lemmas and discussions above. 

Before we proceed, let us fix some notation that will be used throughout the section. Let $F = \{e,e'\}$. Let $\rr$ be the root of $\ft(st)$. If the path $\pr_\rr =st$ does not contain any faulty edge, then our algorithm will just return the path at the root $\rr$, that is $\pr_\rr$. So, let us assume that $e$ is the edge closest to $t$ on $\pr_\rr$. Let $xy$ be the segment containing $e$. Let $R = \stf$. 

We first look at the simpler case when $R$ does not pass through $xy$. Consider the call $\quDtwo(t,\ro(\ft(st)))$. Let $\xx$ be the child of the root $\rr$ where we removed segment $xy$. Note that $\xx$ must exist as there is a 2-decomposable path $R$ that avoids the segment $xy$. We then recurse with $\quDtwo(t,\xx)$. At node $\xx$, the associated graph $G_\xx$ contains at most one fault, and $R$ survives in $\xx$. The reader can check that the working of $\quDtwo$ in \Cref{querynormal} matches that in \Cref{query} for the single fault case. Similar to our analysis in single fault case, we can show that $\quDtwo(t,\ro(\ft(st))$ returns a $(1+k^5\epsilon)$-approximation of $R$\label{R_does_not_pass_xy} 
\ifnum\myone=1(see \Cref{proof of 7.1} for complete discussion)
\else
See the full version of the paper below for the details. 
\fi

In the rest of the section, we handle the harder case, when $R$ passes through $xy$. Since $R$ passes through $xy$, it may intersect $xy$ either above or below $e$ (or both). If it intersects $xy$ both above and below $e$, we will give preference to its intersection below $e$ in the analysis. Let us now look at these two subcases in detail.

\subsection{\texorpdfstring{$R$ intersects $xy$ below $e$ at the vertex $z$}{R intersects xy below e at vertex z}}
\label{case:2a} 

		This case is identical to the case in \Cref{item:R_intersect_below_e}. We briefly mention the major arguments.   In 
            \ifnum\myone=1
            \Cref{querynormal:10}
        \else
            Line \ref{querynormal:10}
        \fi we set 
        \begin{align*}
        	\pa &= |\pr_\rr[v,t]| + \quDtwo(v, \ro(\ft(sv)))\\
        	\intertext{Using \Cref{eq:one} and setting $Q=\pr_\rr[v,t]$, $S = xy$, we get:}
        	& = |Q| + |S[z,v]| +  \quDtwo(v, \ro(\ft(sv)))
        \end{align*}

        As in \Cref{item:R_intersect_below_e}, we can show that $\pa$ is a $(1+k^5\ep)$-approximation of $|\stf|$. The approximation is weaker because, unlike in \Cref{lem:k}, the induction in \Cref{lem:kk} uses a larger approximation ratio.

		\subsection{\texorpdfstring{$R$ intersects $xy$ above $e$ at the vertex $z$}{R intersects xy above e}}
        \label{case:2b}
		
        It means that $R$ does not intersect $xy$ below $e$. To obtain an appropriate approximation, we compare $R$ with $st \diamond e$.  
        The path $st \diamond e$ can behave in one of the two ways with respect to the segment $xy$: (1) $st \diamond e$ intersects $xy$ above $e$ or (2) below $e$ or avoids the whole $xy$. We now handle both the cases separately.

        \subsubsection{\texorpdfstring{$st \dia e$ intersects the segment $xy$ above the edge $e$}{st dia e intersects segment xy above edge e}}\label{case:2b-alpha1}

             Here, we follow our arguments from \Cref{item:R intersect above e}. 
             Assume that the detour length  $st \dia e$ is in $\left[(\Ep)^{i-1}, (\Ep)^i\right]$ where $i \ge 1$. There exists a 1-decomposable detour child $\rr_i$ of $\rr$ such that the detour of  $\pr_\rri$ starts above $e$ has detour length at most $(1+\epsilon)^i$. This detour child must exists as $st \dia e$ itself is one of the candidates.  This detour child is constructed in 
        \ifnum\myone=1
            \Cref{line:line11}
        \else
            Line \ref{line:line11}
        \fi. As in \Cref{item:R intersect above e}, we can show that $\pr_\rri$ is a $(\Ep)$-approximation of $st\dia e$. Thus, we have the following lemma:
             \begin{lemma}
             \label{lem:boundry}
                $|\pr_\rri| \le(\Ep)|st \dia e| \le (\Ep) |R|$
             \end{lemma}
             
             In the above lemma, the second inequality is due to the fact that $|R| \ge |st \dia e|$. Now, if $\pr_\rri$ does not contain the second faulty edge $e'$, then we will simply return the path $\pr_\rri$ (in    
        \ifnum\myone=1
            \Cref{qn:15})
        \else
            Line \ref{qn:15})
        \fi and we will get a $(1+\epsilon)$ approximate path.      
            
            So, the only remaining case is when $\pr_\rri$  contains another faulty edge $ e'=(u',v')$. Let $pq$ be the segment in $\pr_\rri$ that contains $e'$. Again there are three cases:  $R$ can either intersect with $\pr_\rri$ above $e'$ on $pq$, or below $e'$  on $pq$ or totally avoid the segment $pq$.

        \begin{enumerate}[label=($\beta_\arabic*$)]

        \item\label{case:2b-alpha1-i} 
        $R$ does not intersect $pq$ or the detour of $R$ starts from $xy$ (See \Cref{Does_not_pass_through_pq})

        While processing $\rr_i$ in $\mn$, we constructed a 2-decomposable detour child $\xx_j$ of $\rr_i$ such that the detour of $\pr_{\xx_j}$  starts closest to $x$ on $xy$ and detour length is at most $(1+\ep)^i$. Such a detour must exist since $R$ itself is a valid candidate. Let the detour of $R$ starts at $z$ on $xy$, and let the detour of $\pr_{\xx_j}$ starts at $w$. Then $w$ lies above $z$ on the $xy$ segment, as $R$ is one possible candidate for $\pr_{\xx_j}$. Moreover, $\pr_{\xx_j}[w,t]$ is a $(\Ep)$-approximation of $R[z,t]$. As $\pr_{\rr_i}$ contains another faulty edge $e'$, and $\pr_{\xx_j}$ avoids the segment $pq$ containing $e'$, the path $\pr_{\xx_j}$ avoids the entire fault set $F$. Let $r$ be the vertex in $\{u,u'\}$ that is closest to $z$ (See \Cref{Does_not_pass_through_pq} for an illustration of these two cases). Let us look at the simpler case when $z$ is close to $u'$ (See \Cref{fig:r=u'}). In
        \ifnum\myone=1
            \Cref{querydetourone:5}
        \else
            Line \ref{querydetourone:5}
        \fi of $\quDone$, we set $\pa$ to:
        \begin{figure}
    \centering
    \begin{minipage}[t]{0.45\textwidth}
        \centering
\begin{tikzpicture}[scale=0.8, every node/.style={font=\small}]
        \tikzstyle{point} = [circle, fill=black, inner sep=1.5pt]
        \tikzstyle{segment} = [thick]

        \draw[thick] (0,7) -- (0,0);

        \draw[segment] (0,6) -- (0,2);
        
        \draw[red, ultra thick] (0,4) -- (0,3) node[midway, left] {\scriptsize $e$};
        \draw[orange, thick] (0,1.5) arc[start angle=-90, end angle=90, radius=2]node[pos=0.8, below=-12, sloped, black, rotate=360] {\scalebox{0.7}{$\pr_{\rr_i}$}};

        \draw[pink, ultra thick] (0,1) arc[start angle=90, end angle=-90, radius=-2]node[pos=0.5, below=-12, sloped, black, rotate=360] {\scalebox{0.7}{$\pr_{\xx_j}$}};
    
        \draw[magenta, line width=0.8mm, opacity=0.7]
        (0.7,1.6) .. controls (1.2,3) and (0.5,4.2) .. (0,4.5)node[pos=0.5, below=-12, sloped, black, rotate=360] {\scalebox{0.7}{$R$}};

        \draw[pink, line width=0.8mm, opacity=0.7] (0,1.5) -- (0,0);

        \draw[pink, line width=0.8mm, opacity=0.7] (0,4.5) -- (0,7);

        \node[point, pink, label=right:{\scalebox{0.7}{$w$}}] at (0,5) {};
        \node[point, blue, label=left:{\scalebox{0.7}{$z$}}] at (0,4.5) {};
        \node[point, red, label=left:{\scalebox{0.7}{$u$}}] at (0,4) {};
        \node[point, red, label=left:{\scalebox{0.7}{$v$}}] at (0,3) {};
        \node[point, red, label=right:{\scalebox{0.7}{$u'$}}] at (1.95,4) {};
        \node[point, red, label=right:{\scalebox{0.7}{$v'$}}] at (1.7,2.5) {};
        \node[point, label=right:{\scalebox{0.7}{$p$}}, teal] at (1.7,4.6) {};
        \node[point, teal, label=left:{\scalebox{0.7}{$x$}}] at (0,6) {};
        \node[point, teal, label=left:{\scalebox{0.7}{$y$}}] at (0,2) {};
        \node[point, label=above:$s$] at (0,7) {};
        \node[point, label=below:$t$] at (0,0) {};
        \node[point, label=below:{\scalebox{0.7}{$q$}},teal] at (1.2,1.9) {};
        \draw[red, ultra thick] (1.7,2.5) to[bend right=20] node[midway, left] {\scriptsize $e'$} (1.95,4);
        \end{tikzpicture}
\subcaption{when $r=u$}
\label{fig:r=u}
\end{minipage}
\hfill
\begin{minipage}[t]{0.45\textwidth}
    \centering
    \begin{tikzpicture}[scale=0.8, every node/.style={font=\small}]
    \tikzstyle{point} = [circle, fill=black, inner sep=1.5pt]
    \tikzstyle{segment} = [thick]

    \draw[thick] (0,7) -- (0,0);

    
    \draw[red, ultra thick] (0,2) -- (0,3) node[midway, left] {\scriptsize $e$};
    \draw[red, ultra thick] (0,4) -- (0,5) node[midway, left] {\scriptsize $e'$};
    \draw[orange, thick] (0,0.7) arc[start angle=-90, end angle=90, radius=1.5]node[pos=0.5, below=12, sloped, black, rotate=180] {\scalebox{0.9}{$\pr_{\rr_i}$}};

    \draw[magenta, ultra thick] (0,0.5) arc[start angle=90, end angle=-90, radius=-2.4]node[pos=0.5, below=-12, sloped, black, rotate=360] {\scalebox{0.7}{$R$}};

    \draw[pink, ultra thick] (0,6) arc[start angle=90, end angle=-90, radius=2.8]node[pos=0.5, below=-12, sloped, black, rotate=360] {\scalebox{0.9}{$\pr_{\xx_j}$}};
    
    \node[point, label=above:$s$] at (0,7) {};
    \node[point, label=below:$t$] at (0,0) {};
    \node[point, magenta, label=right:{\scalebox{0.7}{$z$}}] at (0,5.3) {};
    \node[point, red, label=left:{\scalebox{0.7}{$u$}}] at (0,3) {};
    \node[point, red, label=left:{\scalebox{0.7}{$v$}}] at (0,2) {};
    \node[point, pink, label=left:{\scalebox{0.7}{$w$}}] at (0,6) {};
    \node[point, red, label=left:{\scalebox{0.7}{$u'$}}] at (0,5) {};
    \node[point, red, label=left:{\scalebox{0.7}{$v'$}}] at (0,4) {};
    \node[point, label=right:{\scalebox{0.7}{$p$}}, teal] at (0,5.6) {};
    \node[point, teal, label=right:{\scalebox{0.7}{$x$}}] at (0,6.7) {};
    \node[point, teal, label=left:{\scalebox{0.7}{$y$}}] at (0,1) {};
    \node[point, label=below:{\scalebox{0.7}{$q$}},teal] at (0.9,1) {};

    \end{tikzpicture}
\subcaption{when $r=u'$}
\label{fig:r=u'}
\end{minipage}
\caption{$\stf$ does not intersect $pq$ and the detour starts from $xy$}
\label{Does_not_pass_through_pq}
\end{figure}
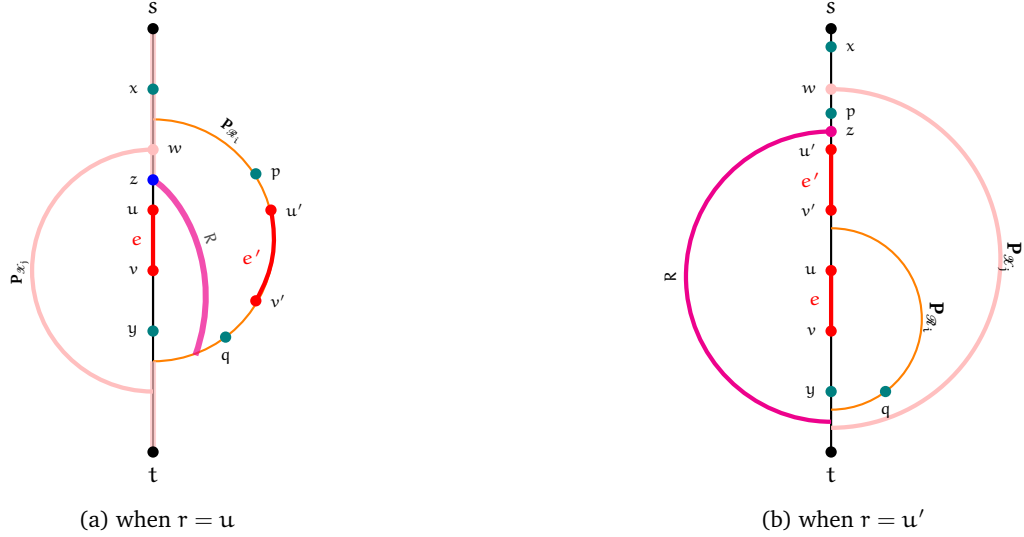
    \begin{align*}
            \pa &= |\pr_{\xx_j}[w,t]| +|\pr_\rr[w,r]| + \quDtwo(r, \ro(\ft(sr))) \\
            \intertext{As mentioned above, $\pr_{\xx_j}[w,t]$ is a $(1+\ep)$-approximation of $R[z,t]$. Thus,}
            &\le (\Ep)|R[z,t]| +|\pr_\rr[w,r]| + \quDtwo(r, \ro(\ft(sr))) \\
            \intertext{Now using \Cref{lem:vless}, we can say that $|sr \dia F| \leq |st\dia F|$}
            \intertext{Thus, $r \in L$ and by induction, $\quDtwo(r, \ro(\ft(sr)))$ is a $(1+(k-1)^5\ep)$-approximation of $|sr \dia F|$.} 	
            	&\le (\Ep)|R[z,t]| +|\pr_\rr[w,r]| + (1+(k-1)^5\ep) |sr \dia F| \\
			\intertext{If $r = u'$, then $e'$ lies on the $st$ path above $e$ (See \Cref{fig:r=u'}). Thus, $|su' \dia F| = |su'| \le |sz| \le  |sz \dia F|$. Note that we are crucially using the fact that there is no faulty edge above $e'$, or in other words, there are only two faults. If $r=u$ (See \Cref{fig:r=u}), then using triangle inequality, $|su \dia F| \le |sz \dia F| + |zu \dia F|$. But $zu$ lies in the segment $xy$ and does not contain any faults. Thus, $|su \dia F| \le |sz \dia F| + |zu| \le |sz \dia F| + |xy|$ where the last inequality is due to the fact that $zu$ is a subpath of the segment $xy$. Combining both the cases, $r=u$ and $r=u'$, we get $|sr \dia F| \le  |sz \dia F| + |xy|$. Adding it in the above equation gives:}
			&\leq (\Ep)|R[z,t]| +|\pr_\rr[w,r]| + (1+(k-1)^5\ep) (|sz \dia F| + |xy|) \\
			\intertext{Since $\pr_\rr[w,r]$  lies on the segment $xy$, their length is $\le |xy|$. Thus, we get: }
			&\leq (\Ep)|R[z,t]| +(2+(k-1)^5\ep)|xy| + (1+(k-1)^5\ep) |sz \dia F| \\
			\intertext{$xy$ lies on the path $\pr_\rr$. Using \Cref{lem:chechik}, $|xy| \le \ep |\pr_\rr| \le \ep |R|$.}
			&\le (\Ep)|R[z,t]| +\ep(2+(k-1)^5\ep)|R| + (1+(k-1)^5\ep) |sz \dia F| \\
			\intertext{Since $R$ passes through $z$, $R[s,z] = |sz \dia F|$.}
			&= (\Ep)|R[z,t]| +\ep(2+(k-1)^5\ep)|R| + (1+(k-1)^5\ep) |R[s,z]| \\
			&\le  (1+(k-1)^5\ep)|R| +\ep(2+(k-1)^5\ep)|R| \\
			&\le (1+k^5\ep) |R|
		\end{align*}
            
                \item\label{case:2b-alpha1-ii} $R$ intersects $\pr_{\rr_i}$ above $e'$ in $pq$
                            (See \Cref{Intersecting_above_e}).

            	\begin{figure}[t]
                	\centering
                    \begin{tikzpicture}[scale=0.8, every node/.style={font=\small}]
        \tikzstyle{point} = [circle, fill=black, inner sep=1.5pt]
        \tikzstyle{segment} = [thick]

        \draw[thick] (0,7) -- (0,0);

        \draw[segment] (0,6) -- (0,2);
        
        \draw[red, ultra thick] (0,4) -- (0,3) node[midway, left] {\scriptsize $e$};
        \draw[orange, thick] (0,1.5) arc[start angle=-90, end angle=90, radius=1.5];
        \draw[magenta, line width=0.8mm, opacity=0.7] (0,0.8) arc[start angle=-120, end angle=57, radius=1.4]node[pos=0.5, below=12, sloped, black, rotate=180] {\scalebox{0.7}{$R$}};
        \draw[pink, line width=0.8mm, opacity=0.7] (0,0.8) -- (0,0);
        \draw[pink, line width=0.8mm, opacity=0.7] (0,4.5) -- (0,7);

        \node[point, magenta, label=right:{\scalebox{0.7}{$w$}}] at (1.5,3.2) {};
        \node[point, label=left:{\scalebox{0.7}{$v$}},red] at (0,3) {};
        \node[point, label=left:{\scalebox{0.7}{$u$}},red] at (0,4) {};
        \node[point, label=below:{\scalebox{0.7}{$q$}},teal] at (0.7,1.7) {};
        \node[point, label={\scalebox{0.7}{$p$}}, teal] at (1.1,4) {};
        \node[point, teal, label=left:{\scalebox{0.7}{$x$}}] at (0,6) {};
        \node[point, teal, label=left:{\scalebox{0.7}{$y$}}] at (0,2) {};
        \node[point, red, label=left:{\scalebox{0.7}{$v'$}}] at (1.27,2.2) {};
        \node[point, label=above:$s$] at (0,7) {};
        \node[point, label=below:$t$] at (0,0) {};
        \node[point, red, label=left:{\scalebox{0.7}{$u'$}}] at (1.48,2.8) {};
        \draw[red, ultra thick](1.27,2.2) to[bend right=12]
        node[midway, right] {\scriptsize $e'$}
        (1.48,2.8);

        \end{tikzpicture}
                    \captionof{figure}{$\stf$ intersects $pq$ above $e'$.}
                \label{Intersecting_above_e}
                \end{figure}
            
                This case mirrors the above case. 
                While processing $\rr_i$, we constructed a 2-decomposable detour child in  $\yy_j$ of $\rr_i$ such that the detour of $\pr_{\yy_j}$ 
                starts as high as possible on $pq$ and whose detour length is at most $(1+\ep)$ times 
                that of the detour of $R$. The remaining steps are similar to \ref{case:2b-alpha1-i} above -- the only difference is that $w$ lies on the segment $pq$ instead of segment $xy$. 
                Hence, this case also yields a $(1+k^5\ep)$-approximation of $|R|$.
            
            \item\label{case:2b-alpha1-iii} $R$ intersect $\pr_\rri$ somewhere below $e'$ say at a vertex $z$ in the segment $pq$
            (See \Cref{fig:Intersecting_below_e})
    
    \begin{figure}[h!]
    	\centering
        \begin{tikzpicture}[scale=0.8, every node/.style={font=\small}]
    \tikzstyle{point} = [circle, fill=black, inner sep=1.5pt]
    \tikzstyle{segment} = [thick]

    \draw[thick] (0,7) -- (0,0);
    
    \draw[segment] (0,6) -- (0,2);
    \draw[red, ultra thick] (0,4) -- (0,3) node[midway, left] {\scriptsize $e$};
    \draw[orange, thick] (0,1.5) arc[start angle=-90, end angle=90, radius=1.75];

    \draw[magenta, line width=0.8mm, opacity=0.7]
        (1.4,2.2) .. controls (1.2,3) and (0.5,4.2) .. (0,4.5)
        node[pos=0.5, below=-12, sloped, black] {\scalebox{0.7}{$R$}};

    \draw[pink, line width=0.8mm, opacity=0.7] (0,1.5) -- (0,0);
    \draw[pink, line width=0.8mm, opacity=0.7] (0,4.5) -- (0,7);

    \node[point, magenta, label=below:{\scalebox{0.7}{$z$}}] at (1.4,2.2) {};
    \node[point, red, label=left:{\scalebox{0.7}{$u$}}] at (0,4) {};
    \node[point, magenta, label=left:{\scalebox{0.7}{$w$}}] at (0,4.5) {};
    \node[point, red, label=left:{\scalebox{0.7}{$v$}}] at (0,3) {};
    \node[point, label=right:{\scalebox{0.7}{$p$}}, teal] at (1.1,4.6) {};
    \node[point, label=below:{\scalebox{0.7}{$q$}}, teal] at (1.1,1.9) {};
    \node[point, teal, label=left:{\scalebox{0.7}{$x$}}] at (0,6) {};
    \node[point, teal, label=left:{\scalebox{0.7}{$y$}}] at (0,2) {};
    \node[point, red, label=right:{\scalebox{0.7}{$u'$}}] at (1.6,4) {};
    \node[point, red, label=right:{\scalebox{0.7}{$v'$}}] at (1.6,2.5) {};
    \node[point, label=above:$s$] at (0,7) {};
    \node[point, label=below:$t$] at (0,0) {};

    \draw[red, ultra thick] (1.6,2.5) to[bend right=20] node[midway, right] {\scriptsize $e'$} (1.6,4);
\end{tikzpicture}
        \captionof{figure}{$\stf$ intersects $pq$ below $e'$.}
    \label{fig:Intersecting_below_e}
    \end{figure}

            Let $e' = (u',v')$, where $v'$ lies closer to $t$ on $\pr_\rri$.  
            
            In     
            \ifnum\myone=1
                \Cref{querydetour:6}
            \else
                Line \ref{querydetour:6}
            \fi of $\quDone$, we set 
            
    		\begin{align*}
    			\pa &= |\pr_\rri[v',t]| + \quDtwo(v', \ro(\ft(sv'))) \\
    				&= |\pr_\rri[v',z]| + |\pr_\rri[z,t]| + \quDtwo(v', \ro(\ft(sv')))\\
                \intertext{Let $S = pq$ and $Q = \pr_\rri[z,t]$. Since $\pr_{\rr_i}[v',z]$ completely lies on segment $pq$, $S[v',z] = \pr_\rri[v',z]$. Thus, we get:}\\
                &= S[v',z] + |Q| + \quDtwo(v', \ro(\ft(sv'))) \\
    		\end{align*}

    Thus, we can apply \Cref{lem:belowe}. We just need to find the approximation that the three components of $\pa$ provide. Let's go over them one by one.
   
		\begin{itemize}
			\item $Q = \pr_{\rr_i}[z,t]$
					
			After $R$ intersects segment $pq$ at $z$, we claim that $R$ must continue along the path $\pr_{\rr_i}$. To justify this, note the following: First, $R[z,t]$ cannot be longer than $\pr_{\rr_i}[z,t]$. If it were, we could replace $R[z,t]$ with $\pr_{\rr_i}[z,t]$ and obtain a path shorter than $R$, contradicting the minimality of $R$.  Second, $R[z,t]$ cannot be shorter than $\pr_{\rr_i}[z,t]$. If it were, then while processing $\rr_i$, we could replace $\pr_{\rr_i}[z,t]$ with $R[z,t]$, resulting in a smaller detour at $\rr_i$, contradicting the minimality of $\pr_\rri$. Therefore, by setting $Q$ to follow $\pr_{\rr_i}[z,t]$, we achieve a $1$-approximation of $|R[z,t]|$.

			\item $S = pq$
			
			By \Cref{lem:chechik}, the segment length satisfies $|pq| \le \ep|\pr_\rri|$. But by \Cref{lem:boundry}, $|\pr_\rri|$ is a $(1+\ep)$ of $|R|$. Hence,  $S$ gives a $\overbrace{\ep(1+\ep)}^{\beta}$-approximation of $|R|$.

			\item $\quDtwo(v', \ro(\ft(sv')))$
			
			Similar to \Cref{lem:vless}, we can show that $|sv' \dia F| < |st \dia F|$. Thus, $v' \in L$ and in the worst case, it lies just before $t$ in $L$, or it is the $(k-1)$-th element in the list. Using induction hypothesis  $\quDtwo(v', \ro(\ft(sv')))$ returns a $\overbrace{(1+(k-1)^5\ep)}^{\alpha}$-approximation of $|sv' \dia F|$.
		\end{itemize}
        
Substituting  $\alpha = 1+(k-1)^5\ep$, and $\beta = \ep(1+\ep)$ into \Cref{lem:belowe}, we get $\pa \le (1 + k^5\ep)|R|$ as required.

        \end{enumerate}

        \subsubsection{\texorpdfstring{$st \dia e$ intersects the segment $xy$ below the edge $e$ or $st \dia e$ avoids the whole segment $xy$}{st dia e intersects xy below e or avoids xy}}
        \label{case:2}

    As in Case \ref{case:2b-alpha1} , now we cannot claim that there is a 1-decomposable detour child $\rr_i$ of $\rr$ such that the detour of  $\pr_\rri$ starts above $e$. This is because $st \dia e$ itself does not intersect the segment $xy$ above $e$, so such a 1-decomposable detour may not even exist. We now use a different strategy for which we have created different detour children in 
    \ifnum\myone=1
        \Cref{line:2-decompose}
    \else
        Line \ref{line:2-decompose}
    \fi of \Cref{alg:constructftgeneral}.

    In 
    \ifnum\myone=1
        \Cref{line:2-decompose}
    \else
        Line \ref{line:2-decompose}
    \fi of \Cref{alg:constructftgeneral}, we created a 2-decomposable detour child $\rri$ such that $\pr_{\rr_i}$  avoids one edge of the segment $xy$, 
    while the second faulty edge may lie anywhere in the graph $G-xy$, and its detour starts on the segment $xy$ and its detour length is at most $(\Ep)$ the detour of $R$. Such a path $\pr_\rri$ must exist as $R$ is itself a candidate.  
    
    If $\pr_{\rr_i}$ does not contain the second faulty edge $e'$, we simply return $\pr_{\rr_i}$. As in \Cref{lem:boundry}, it follows that $\pr_{\rr_i}$ is an $(\Ep)$-approximation of $R$. Hence, in the remainder, we assume that $\pr_{\rr_i}$ contains the faulty edge $e'$. 

Let us again set the context and notation that will be used throughout the rest of the section. We have found a path $\pr_{\rr_i}$ whose detour starts above $e$ on $xy$, and joins below $xy$. It contains the second faulty edge $e'$, We assume that the edge $e'$ belongs to the segment $pq$ on $\pr_\rri$. The path $R$ can interact with $\pr_{\rr_i}$ in the following ways:

\textbf{Case 1: $R$ avoids $pq$ entirely:} This implies that the detour of $R$ starts from the segment $xy$ and does not intersect the segment $pq$. While processing $\rr_i$, we constructed a 2-decomposable detour child $\xx_j$ of $\rr_i$ such that the detour of $\pr_{\xx_j}$ detour starts as high as possible on $xy$ and whose detour length is at most $(1+\ep)$ times that of the detour of $R$. Such a detour must exist since $R$ itself is a valid candidate. Let the detour of $R$ start at $z$ on $xy$, and let the detour of $\pr_{\xx_j}$ start at $w$. This is a similar case of \ref{case:2b-alpha1-i}. We know that  $\pr_{\xx_j}[w,t]$ is a $(\Ep)$-approximation of $R[z,t]$. The remaining analysis is same as of \ref{case:2b-alpha1-i}, and yielding a $(1+k^5\ep)$-approximation of $|R|$.

\medskip
\noindent
\textbf{Case 2: $R$ intersects $\pr_{\rr_i}$ above $e'$ on the segment $pq$:} If $R$ intersects $\pr_{\rr_i}$ above $e'$ on $pq$, then the detour of $R$ starts from the segment $pq$ and our the argument is analogous to \ref{case:2b-alpha1-ii}. While processing $\rr_i$ in $\mn$, we constructed a detour child $\yy_j$ of $\rr_i$ whose detour starts closest to $p$ on $pq$ and whose detour length is at most $(1+\ep)$ times that of the detour of $R$ and it joins the $st$ path below the segment $pq$. Similar to the case of \ref{case:2b-alpha1-ii}, we will get a $(1+k^5\ep)$-approximation of $|R|$.

\medskip
\noindent
\textbf{Case 3: $R$ intersects $\pr_{\rr_i}$ 
below $e'$ on the segment $pq$:} 

This is the hardest case for us. We now make some observations about the path $R$ and $st \dia e$.

    \begin{lemma}
    \label{lem:second_faulty_edge}
        $st \dia e$ must contain the second faulty edge $e'$. 
    \end{lemma}
    \begin{proof}
        Assume for contradiction that $st \dia e$ does not contain $e'$. Then $R$ should have followed $st \dia e$ i.e $R = st \dia e$. But by our assumption in this case, $st \dia e$ does not intersect segment $xy$ above $e$ but $R$ does. Thus, we arrive at a contradiction.
    \end{proof}

    Let us see now how $R$ may intersects with $st \dia e$. We show that $R$ cannot intersects $st \dia e$ above the second faulty edge $e'$ and if it intersects with $st \dia e$ below $e'$, then it must follow $st \dia e$ till $t$.
    
    \begin{lemma}
    \label{lem:decompose}
    $R$ can not intersect $st \dia e$  above the second faulty edge $e'$.  If $R$ intersects $st \dia e$ below the second faulty edge $e'$, say at a point $z$, then $R$ coincides with $st \dia e$ from $z$ to $t$, i.e., 
    \[
    R[z,t] = (st \dia e)[z,t].
    \]
    \end{lemma}
    \begin{figure}[H]
    \centering
    \begin{minipage}[t]{0.40\textwidth}
        \centering
        \begin{tikzpicture}[scale=0.8, every node/.style={font=\small}]
        \tikzstyle{point} = [circle, fill=black, inner sep=1.5pt]
        \tikzstyle{segment} = [thick]

        \draw[thick] (0,7) -- (0,0);

        \node[point, label=above:$s$] at (0,7) {};
        \node[point, label=below:$t$] at (0,0) {};

        \draw[segment] (0,4.5) -- (0,2.5);
        \draw[pink, thick] (0,2) -- (0,6);
        
        \draw[red, ultra thick] (0,4.5) -- (0,3.5) node[midway, left] {\scriptsize $e$};
        
        \draw[orange, thick] (0,3) arc[start angle=-90, end angle=90, radius=1.75] node[pos=0.8, below=-2, sloped, black, rotate=360] {\scalebox{0.7}{$st \dia e$}};

        \draw[blue, thick] (0,5.5) .. controls (0.8,4.8) .. (1.7,5.2)node[pos=0.5, below=-12, sloped, black, rotate=360] {\scalebox{0.7}{$R$}};

        \draw[blue, thick] (0,1.5) .. controls (0.5,2) .. (0.95,3.3);

        \draw[blue, thick](0.95,3.3).. controls (1.8,3.9) and (1.8,4.7).. (1.7,5.2);
       
        \node[point, label=right:{\scalebox{0.7}{$z$}},blue] at (1.7,5.2) {};
        \node[point, label=left:{\scalebox{0.7}{$v$}},red] at (0,3.5) {};
        \node[point, label=left:{\scalebox{0.7}{$u$}},red] at (0,4.5) {};
        \node[point, label=below:{\scalebox{0.7}{$q$}},teal] at (0.6,3.1) {};
        \node[point, label=right:{\scalebox{0.7}{$p$}}, teal] at (1.45,5.7) {};
        \node[point, teal, label=left:{\scalebox{0.7}{$x$}}] at (0,6) {};
        \node[point, teal, label=left:{\scalebox{0.7}{$y$}}] at (0,2) {};
        \draw[red, ultra thick](1.2,3.5) to[bend right=16] node[midway, right] {\scriptsize $e'$}(1.68,4.25);
        \node[point, red, label=left:{\scalebox{0.7}{$v'$}}] at (1.2,3.5) {};
        \node[point, red, label=left:{\scalebox{0.7}{$u'$}}] at (1.68,4.25) {};
\end{tikzpicture}
\subcaption{$\stf$ when intersect $st \dia e$ intersect above the edge $e'$}
\label{Intersecting_above_e'}
\end{minipage}
\begin{minipage}[t]{0.40\textwidth}
        \centering
        \begin{tikzpicture}[scale=0.8, every node/.style={font=\small}]
    \tikzstyle{point} = [circle, fill=black, inner sep=1.5pt]
        \tikzstyle{segment} = [thick]

        \draw[thick] (0,7) -- (0,0);

        \node[point, label=above:$s$] at (0,7) {};
        \node[point, label=below:$t$] at (0,0) {};

        \draw[segment] (0,4.5) -- (0,2.5);
        \draw[pink, thick] (0,2) -- (0,6);
        
        \draw[red, ultra thick] (0,4.5) -- (0,3.5) node[midway, left] {\scriptsize $e$};
        
        \draw[orange, thick] (0,3) arc[start angle=-90, end angle=90, radius=1.75] node[pos=0.8, below=-2, sloped, black, rotate=360] {\scalebox{0.7}{$st \dia e$}};
        

        \draw[blue, thick] (0,5.5) .. controls (0.4,4.2) .. (1.4,3.7)node[pos=0.8, below=-2, sloped, black, rotate=360] {\scalebox{0.7}{$R$}};
        
        \node[point, label=left:{\scalebox{0.7}{$v$}},red] at (0,3.5) {};
        \node[point, label=left:{\scalebox{0.7}{$u$}},red] at (0,4.5) {};
        \node[point, label=below:{\scalebox{0.7}{$q$}},teal] at (0.6,3.1) {};
        \node[point, label=right:{\scalebox{0.7}{$p$}}, teal] at (1.45,5.7) {};
        \node[point, label=right:{\scalebox{0.7}{$z$}}, blue] at (1.4,3.7) {};
        \node[point, teal, label=left:{\scalebox{0.7}{$x$}}] at (0,6) {};
        \node[point, teal, label=left:{\scalebox{0.7}{$y$}}] at (0,2) {};
        \draw[red, ultra thick](1.68,4.25) to[bend right=16] node[midway, right] {\scriptsize $e'$}(1.68,5.25);
        \node[point, red, label=left:{\scalebox{0.7}{$u'$}}] at (1.68,5.25) {};
        \node[point, red, label=left:{\scalebox{0.7}{$v'$}}] at (1.68,4.25) {};
\end{tikzpicture}
\subcaption{\scalebox{0.9}{$\stf$ pass through the segment $pq$}}
        \label{Intersecting_below_e'}
    \end{minipage}
    \hfill
    \caption{$\stf$ when intersect $st \dia e$ on the segment $pq$ containing the edge $e'$}
    \label{fig:two_edge_faultp}
\end{figure}
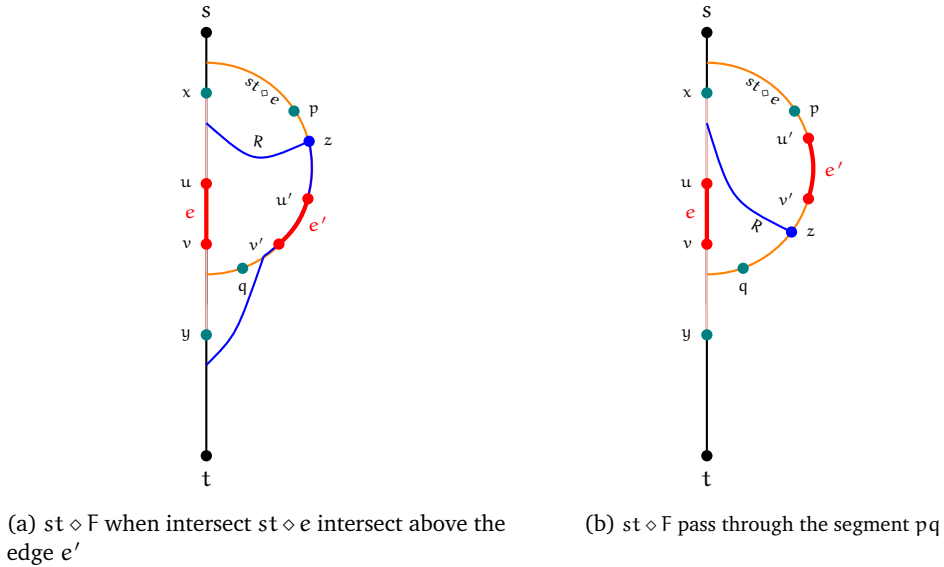
    \begin{proof}
        Assume for contradiction that $R$ intersects $st \dia e$ at $z$ above the second faulty edge $e'$ (see
        \ifnum\myone=1
            \Cref{Intersecting_above_e'})
        \else
            Figure \ref{Intersecting_above_e'})
        \fi. Then there are two different shortest paths from $s$ to $z$ avoiding $F = \{e,e'\}$. One is  $(st \dia e)[s,z]$ and the other is $R[s,z]$, a contradiction since all shortest paths, even avoiding two edges, are unique.
        
    Now, if $R$ intersects $st \dia e$ below the second faulty edge $e'$, say at a vertex $z$ (see 
    \ifnum\myone=1
        \Cref{Intersecting_below_e'})
    \else
        Figure \ref{Intersecting_below_e'})
    \fi . Then the path $R$ must follow $(st \dia e)[z,t]$, as $(st \dia e)[z,t]$ avoids $F$. Hence, $R[z,t] = (st \dia e)[z,t]$, which completes the proof.
    \end{proof}

By our assumption, $\pr_{\rr_i}$ contains $e'$. But using \Cref{lem:second_faulty_edge}, even $st \dia e$ contains $e'$. This implies that $\pr_\rri$ and $st \dia e$ must be the same after $e'$. We now prove this intuition.
\begin{lemma}\label{pr_equal_st_e}
Let $e' = (u',v')$. Let $v'$ be the vertex closest to $t$ on $st \dia e$. If $\pr_{\rr_i}$ contains the second faulty edge $e'$, then $(st \dia e)[v',t] = \pr_{\rr_i}[v',t]$.
\end{lemma}
\begin{proof}
By the property of the shortest paths, the subpath $(st \dia e)[v',t]$ is also the shortest path from $v'$ to $t$ avoiding $e$. Let us assume that $\pr_\rri$ avoids $\{e_1,e_2\}$. By construction, $e_1$ lies on the segment $xy$ and $e_2$ lies in $G-xy$.  We know that the detour path $\pr_{\rr_i}$ avoids edge $e_2$ of the graph $G-xy$. If the edge $e_2$ does not lie on the subpath $(st \dia e)[v',t]$,  then  $\pr_{\rr_i}$ should have just follow $st \dia e$, that is $(st \dia e)[v',t] = \pr_{\rr_i}[v',t]$.

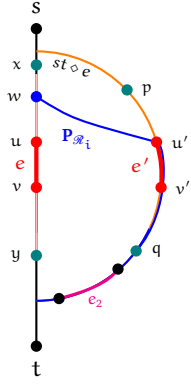
\begin{figure}[t]
\centering
\begin{tikzpicture}[scale=0.6, every node/.style={font=\small}]
        \tikzstyle{point} = [circle, fill=black, inner sep=1.5pt]
        \tikzstyle{segment} = [thick]

        \draw[thick] (0,7) -- (0,0);

        \node[point, label=above:$s$] at (0,7) {};
        \node[point, label=below:$t$] at (0,0) {};

        \draw[segment] (0,4.5) -- (0,2.5);
        \draw[pink, thick] (0,2) -- (0,6);
        
        \draw[red, ultra thick] (0,4.5) -- (0,3.5) node[midway, left] {\scriptsize $e$};

        
        \draw[orange, thick] (0,1) arc[start angle=-90, end angle=90, radius=2.75] node[pos=0.9, below=-2, sloped, black, rotate=360] {\scalebox{0.7}{$st \dia e$}};
        
        \draw[blue, thick] (0,5.5) .. controls (0.8,5) .. (2.65,4.5)node[below, midway] {\scriptsize $\pr_{\rr_i}$};
        
        \draw[blue, thick](2.2,2.1).. controls (3,3) and (2.85,4.5).. (2.65,4.5);





        \draw[blue, thick] (0,3.75) ++(-25:2.75) arc[start angle=-25, end angle=-90, radius=2.75];


        \draw[magenta, line width=0.4mm] (0.5,1.05) .. controls (1.1,1.2) and (1.6,1.4) .. (1.8,1.7) node[midway, below] {\scalebox{0.7}{$e_2$}};

        
        
        \draw[red, ultra thick](2.75,3.5) to[bend right=10] node[midway, left] {\scriptsize $e'$}(2.65,4.5);

        
        \node[point, black] at (0.5,1.05) {};
        \node[point, black] at (1.8,1.7) {};
        \node[point, label=right:{\scalebox{0.7}{$p$}},teal] at (2,5.65) {};
        \node[point, label=right:{\scalebox{0.7}{$q$}},teal] at (2.2,2.1) {};
        \node[point, label=left:{\scalebox{0.7}{$v$}},red] at (0,3.5) {};
        \node[point, label=left:{\scalebox{0.7}{$u$}},red] at (0,4.5) {};
        \node[point, teal, label=left:{\scalebox{0.7}{$x$}}] at (0,6.2) {};
        \node[point, teal, label=left:{\scalebox{0.7}{$y$}}] at (0,2) {};
        \node[point, red, label=right:{\scalebox{0.7}{$u'$}}] at (2.65,4.5) {};
        \node[point, red, label=right:{\scalebox{0.7}{$v'$}}] at (2.75,3.5) {};
        \node[point, blue, label=left:{\scalebox{0.7}{$w$}}] at (0,5.5) {};
        
\end{tikzpicture}
\caption{$e_2$ lies on the $(st\dia e)[v',t]$}
\label{e_2 lies}
\end{figure}

Assume for contradiction that the edge $e_2$ lies on the subpath $(st \dia e)[v',t]$ (See \Cref{e_2 lies}). But then the subpath $(st \dia e)[s,v']$ avoids $\{e_1,e_2\}$ because the detour of $st \dia e$ starts above the segment $xy$. Thus, $st \dia e$ is the shortest path avoiding $\{e_1,e_2,e\}$. We now show that $\pr_\rri$ also the shortest path avoiding $\{e_1,e_2,e\}$. By construction, the detour of $\pr_\rri$ starts above $e$ from the segment $xy$ and joins below segment $xy$. So, $\pr_\rri$ avoids $e$. Thus, $\pr_\rri[s,v']$ avoids $\{e_1,e_2,e\}$. 

By uniqueness of shortest paths, $(st \dia e)[s,v']$ must be same as $\pr_\rri[s,v']$. But we know that the detour of $st \dia e$ start above the segment $xy$ and the detour of $\pr_\rri$ starts on the segment $xy$. Thus, $(st \dia e)[s,v'] \neq \pr_\rri[s,v']$, a contradiction.

\end{proof}

 Now, In 
     \ifnum\myone=1
            \Cref{querydetour:6}
        \else
            Line \ref{querydetour:6}
        \fi of $\quDone$, we set 
            
    		\begin{align*}
    			\pa &= |\pr_\rri[v',t]| + \quDtwo(v', \ro(\ft(sv'))) \\
    				&= |\pr_\rri[z,t]| + |\pr_\rri[v',z]| + \quDtwo(v', \ro(\ft(sv')))\\
                \intertext{Let $S = pq$ and by \Cref{lem:decompose}, we know that if $R$ intersects with $st \dia e$ below the second faulty edge $e'$ at $z$ then,
                $R[z,t]=(st\dia e)[z,t]$.
                Also by \Cref{pr_equal_st_e} we know that if $\pr_\rri$ contains the second faulty edge $e'$ , then
                $(st\dia e)[v',t]=\pr_\rri[v',t]$. Since $z$ lies on the $v'$ to $t$ path on $\pr_\rri$ (and  $st \dia e$), $\pr_\rri[z,t] = (st\dia e)[z,t]$. Thus, $R[z,t]=\pr_\rri[z,t]$. Adding it to the above equation, we get:}
                &= |R[z,t]|+|\pr_\rri[v',z]|  + \quDtwo(v', \ro(\ft(sv')))\\
                \intertext{Setting $Q = R[z,t]$ and $S = pq$. Note that $\pr_\rri[v',z]$ lies entirely in the segment $S$. Thus, $\pr_\rri[v',z] = S[v',z]$. Adding this to the above equation gives:}
                & = Q + S[v'z]+\quDtwo(v', \ro(\ft(sv')))
    		\end{align*}
    			
        The above inequality is of the form \Cref{eq:one}. Thus, we can apply \Cref{lem:belowe} as done in \ref{case:2b-alpha1-iii}. Thus, we get a $(1+k^5\ep)$ approximation of $R$. This completes the proof of our main \Cref{thm:ft-sssp}.
\section{Conclusion and Open Problems}

In this paper, we designed a $\sdo(2)$ with  $o(n^2)$ space and $\tl(1)$ query time. However, our algorithm does not extend to handle three or more edge faults. We explain the reason next.

Let $F = \{e_1, e_2, e_3\}$ be the set of faulty edges on the $st$ shortest path. Suppose $e_2$ and $e_3$ lie on the segment $xy$ of this path, and that $st \dia F$ intersects $st$ between $e_2$ and $e_3$ (see the pink path in the figure \Cref{$3$ edge fault case}). Furthermore, for every vertex $z$ on the subpath $v_1u_2$ of $st$, $|sz \dia F| \gg |\stf|$. Thus, $z$ cannot lie in any path that approximates $\stf$ by a factor of $(\Ep)$.

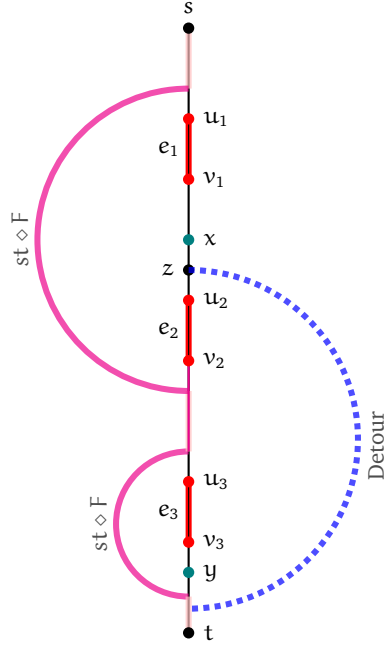
\begin{figure}[t]
    \vspace{-10pt}
    \centering
    \begin{tikzpicture}[scale=0.8, every node/.style={font=\small}]
        \tikzstyle{point} = [circle, fill=black, inner sep=1.5pt]
        \tikzstyle{fault} = [red, line width=0.8mm]

        \draw[thick] (0,10) -- (0,0);

        \draw [magenta, line width=0.8mm, opacity=0.7]
            (0,3) arc [start angle=-90, end angle=90, radius=-1.2]
            node[pos=0.5, below=14, sloped, black, rotate=180] {$\stf$};
		\draw [pink, line width=0.8mm, opacity=0.7] (0,0)--(0,0.6);
		\draw [pink, line width=0.8mm, opacity=0.7] (0,4)--(0,3);
		\draw [pink, line width=0.8mm, opacity=0.7] (0,9)--(0,10);
        \draw [magenta, line width=0.8mm, opacity=0.7]
            (0,9) arc [start angle=-90, end angle=90, radius=-2.5]
            node[pos=0.5, below=14, sloped, black, rotate=180] {$\stf$};
		\node[point, label=left:$z$] at (0,6) {};
        \draw [blue, line width=0.8mm, dotted, opacity=0.7]
            (0,6) arc [start angle=90, end angle=-90, radius=2.8]
            node [pos=0.6, right=10pt, rotate=90, black] {Detour};


        \draw[magenta, thick] (0,3) -- (0,4.5);
        \draw[fault] (0,7.5) -- (0,8.5);
        \draw[fault] (0,4.5) -- (0,5.5);
        \draw[fault] (0,1.5) -- (0,2.5);
        \draw (0,1.5) -- (0,2.5) node[midway, left] {$e_3$};
        \draw (0,4.5) -- (0,5.5) node[midway, left] {$e_2$};
        \draw (0,7.5) -- (0,8.5) node[midway, left] {$e_1$};
        \node[point, label=right:$t$] at (0,0) {};
        \node[point, teal, label=right:$y$] at (0,1) {};
        \node[point, red, label=right:$v_3$] at (0,1.5) {};
        \node[point, red, label=right:$u_3$] at (0,2.5) {};
        \node[point, red, label=right:$v_2$] at (0,4.5) {};
        \node[point, red, label=right:$u_2$] at (0,5.5) {};
        \node[point, red, label=right:$v_1$] at (0,7.5) {};
        \node[point, red, label=right:$u_1$] at (0,8.5) {};
        \node[point, teal, label=right:$x$] at (0,6.5) {};
        \node[point, label=above:$s$] at (0,10) {};

    \end{tikzpicture}
    \caption{ $3$ edge fault case}
    \label{$3$ edge fault case}
    \vspace{-10pt}
\end{figure}

To extend our dual fault strategy to handle three faults, we attempt to find a detour that starts as high as possible on $xy$ and rejoins the $st$ path below $y$. Suppose this detour begins above $e_2$ at a vertex $z$ (see the blue path in the figure). If $e_1$ were absent, the path from $z$ to $s$ along $st$ would suffice to reach $s$. However, the presence of $e_1$ introduces significant complications. By our assumption, |$sz \dia F| \gg |\stf|$, so this detour offers no advantage. In fact, in this configuration, it is crucial to find a detour that starts between $e_2$ and $e_3$ on the $st$ path. However, identifying such a detour is far from straightforward.

We conclude with the following open questions. Can we design a data structure of size $o(n^2)$ that provides a good approximation for the shortest path under three faults? Can this approach be generalised to handle $f$ faults? Moreover, our analysis does not account for preprocessing time -- developing an algorithm with an efficient preprocessing bound remains an interesting direction.

\bibliographystyle{alpha}
\bibliography{jabref.bib}
\appendix
\section{\texorpdfstring{Proof of \Cref{lem:chechik}}{Proof of Lemma chechik}}
\label{Appendix_A}

Assume $u$ is closer to $s$ than $v$, i.e., $|P[s,u]| < |P[s,v]|$.  
Let $i$ be the maximal index such that $|P[s,u]| \ge (1+\epsilon)^i$. There are two cases.  
If $(1+\epsilon)^{i+1} \le |P[s,v]|$, then by definition, both $u$ and $v$ are net points of $P$, so $seg(e,P) = \{e\}$. Otherwise, $(1+\epsilon)^{i+1} > |P[s,v]|$. Then
\vspace{0.5cm}
\[
|seg(e,P)| \le (1+\epsilon)^{i+1} - (1+\epsilon)^i 
= \epsilon(1+\epsilon)^i \le \epsilon |P[s,u]|.
\]
\vspace{0.1cm}

By symmetry, either $seg(e,P) = \{e\}$ or $|seg(e,P)| \le \epsilon |P[v,t]|$.  
If $|seg(e,P)| \le \epsilon |P[v,t]|$, then $|seg(e,P)| \le \epsilon |P[u,t]|$ since $|P[v,t]| < |P[u,t]|$.  Therefore, $seg(e,P) = \{e\} $
or $|seg(e,P)| \le \epsilon \min\{|P[s,u]|, |P[u,t]|\}$,which implies $|seg(e,P)| \le \epsilon |P|$.

\section{\texorpdfstring{Proof of \Cref{lem:insegment}}{Proof of Lemma lca}}
\label{Finding_the_segment_of_an_edge}

Consider an edge $e = (u, v)$ on a shortest path $Q$. Our objective is to identify the segment $xy$ such that $e \in xy$. The endpoints $x$ and $y$ may either belong to the same decomposable subpath $Q_i$ or to two different subpaths $Q_i$ and $Q_j$, where $i \neq j$.

Depending on the length of the decomposable subpath $Q_i$, we distinguish the following cases:

\paragraph{Case 1:}\label{case1_of_appendix_b}
Both endpoints $x$ and $y$ lie on the same decomposed subpath $Q_i$ whose length is at most $\sqrt{n}$. 
Since $|Q_i| \le \sqrt{n}$, we store a dictionary at node $\rr$ using $\tilde{O}(\sqrt{n})$ space, allowing constant-time membership queries for $e \in Q_i$. 
Hence, we can efficiently return the segment $xy$ that contains the edge $e = (u, v)$.

\paragraph{Case 2:}
The segment endpoints $x$ and $y$ lie on the same decomposable path $Q_i$, but the path length satisfies $|Q_i| \ge \sqrt{n}$. 

In this case, we represent $Q_i$ implicitly as a concatenation of two shortest paths, $xz$ and $zy$, where $z \in L$. 
We then examine the tree $T_z$, rooted at $z$, to determine whether the edge $e = (u, v)$ lies within either subpath. Let us assume that $e$ is in $T_z$.

\begin{figure}[h!]
    \centering

    \begin{subfigure}[t]{0.48\textwidth}
        \centering
        \begin{tikzpicture}[>=latex,scale=0.8]
            \tikzstyle{point} = [circle, fill=black, inner sep=1.5pt]

            \node[point, label=below:$x$] at (0,0) {};
            \node[point, label=below:$z$] at (4,0) {};
            \node[point, label=below:$u$] at (6,0) {};
            \node[point, label=below:$v$] at (7,0) {};
            \node[point, label=below:$y$] at (8,0) {};

            \draw[thick, black] (0,0)--(8,0);
            \draw[line width=1.4pt, blue] (6,0)--(7,0)
                node[midway, above, yshift=2pt, blue, scale=0.8] {$e=(u,v)$};
        \end{tikzpicture}
    \end{subfigure}
    \hfill
    \begin{subfigure}[t]{0.48\textwidth}
        \centering
        \begin{tikzpicture}[>=latex,scale=0.8]
            \tikzstyle{point} = [circle, fill=black, inner sep=1.5pt]

            \node[point, label=below:$x$] at (0,0) {};
            \node[point, label=below:$z$] at (4,0) {};
            \node[point, label=below:$u$] at (2,0) {};
            \node[point, label=below:$v$] at (3,0) {};
            \node[point, label=below:$y$] at (8,0) {};

            \draw[thick, black] (0,0)--(8,0);
            \draw[line width=1.4pt, blue] (2,0)--(3,0)
                node[midway, above, yshift=2pt, blue, scale=0.8] {$e=(u,v)$};
        \end{tikzpicture}
    \end{subfigure}

    \caption{Illustration of subpath decomposition of $P_i$ as $xz$ and $zy$, where $z \in L$. Depending on the position of the edge $e=(u,v)$, it may lie in either subpath.}
    \label{fig:segment-detection_cases}
\end{figure}
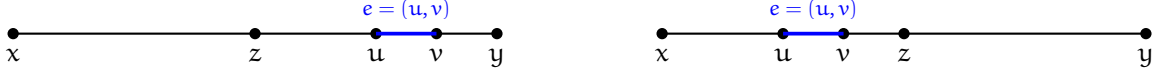

\begin{itemize}
    \item If on $T_z$, $\mathrm{LCA}(v, y) = v$ and $\mathrm{LCA}(u, y) = u$, then $e = (u, v)$ lies in the subpath $zy$. See
    \ifnum\myone=1
        \Cref{fig:segment-detection_cases}
    \else
        Figure \ref{fig:segment-detection_cases}
    \fi 
    \item If on $T_z$, $\mathrm{LCA}(v, x) = v$ and $\mathrm{LCA}(u, x) = u$, then $e = (u, v)$ lies in the subpath $xz$. See 
    \ifnum\myone=1
        \Cref{fig:segment-detection_cases}
    \else
        Figure \ref{fig:segment-detection_cases}
    \fi 
\end{itemize}

\paragraph{Case 3:}
Finally, consider the case where $x$ lies on a decomposable path $Q_i$ and $y$ lies on another decomposable path $Q_j$, where $i \neq j$. 
We represent the path $\pr_\rr$ as 
\[
\pr_\rr = Q_1 \odot e_1 \odot Q_2 \odot e_2 \odot Q_3,
\]
where each $Q_i$ ($i \in \{1,2,3\}$) is a shortest path in $G$.

As illustrated in \Cref{fig:xy_present_in_different_decomposed_path}, suppose the path $P$ is decomposed as
\[
P = Q_i \odot e_1(a,b) \odot Q_k \odot e_2(c,d) \odot Q_j.
\]
Here, $Q_i$ contains one of the segment endpoints $x$, and the path $Q_i$ ends at vertex $a$. The edge $e=(u,v)$ lies in the middle decomposed path $Q_k$, which starts at $b$ and ends at $c$. The other segment endpoint $y$ lies in the decomposed path $Q_j$, which begins at $d$ and ends at $t$. 

If all the decomposed paths have lengths at most $\sqrt{n}$, we can determine whether $x$ lies in $Q_i$, $e=(u,v)$ lies in $Q_k$, and $y$ lies in $Q_j$ using the method described in Case 1.(See
    \ifnum\myone=1
            \Cref{case1_of_appendix_b}
        \else
            \ref{case1_of_appendix_b}
        \fi )
However, if one or more of these paths have length greater than $\sqrt{n}$, we can identify a vertex $z \in L$ for which we have stored the corresponding tree $T_z$. Using $T_z$, we can perform the necessary $\operatorname{LCA}$ queries to determine the locations of $x$, $y$, and $e=(u,v)$. 

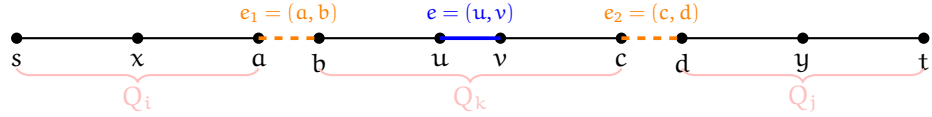
\begin{figure}[h!]
    \centering
    \begin{tikzpicture}[>=latex,scale=0.8]
        \tikzstyle{point} = [circle, fill=black, inner sep=1.5pt]

        \node[point, label=below:$s$] at (0,0) {};
        \node[point, label=below:$x$] at (2,0) {};
        \node[point, label=below:$a$] at (4,0) {};
        \node[point, label=below:$b$] at (5,0) {};
        \node[point, label=below:$u$] at (7,0) {};
        \node[point, label=below:$v$] at (8,0) {};
        \node[point, label=below:$c$] at (10,0) {};
        \node[point, label=below:$d$] at (11,0) {};
        \node[point, label=below:$y$] at (13,0) {};
        \node[point, label=below:$t$] at (15,0) {};

        \draw[thick, black] (0,0)--(4,0);
        \draw[thick, black] (5,0)--(7,0);
        \draw[thick, black] (8,0)--(10,0);
        \draw[thick, black] (11,0)--(15,0);

        \draw[line width=1.4pt, orange, dashed] (4,0)--(5,0)
            node[midway, above, yshift=2pt, orange, scale=0.8] {$e_1=(a,b)$};
        \draw[line width=1.4pt, orange, dashed] (10,0)--(11,0)
            node[midway, above, yshift=2pt, orange, scale=0.8] {$e_2=(c,d)$};

        \draw [decorate,decoration={brace,amplitude=5pt,mirror},pink,thick]
            (0,-0.5) -- (4,-0.5) node[midway,yshift=-0.4cm,pink] {$Q_i$};
        \draw [decorate,decoration={brace,amplitude=5pt,mirror},pink,thick]
            (5,-0.5) -- (10,-0.5) node[midway,yshift=-0.4cm,pink] {$Q_k$};
        \draw [decorate,decoration={brace,amplitude=5pt,mirror},pink,thick]
            (11,-0.5) -- (15,-0.5) node[midway,yshift=-0.4cm,pink] {$Q_j$}; 

        \draw[line width=1.4pt, blue] (7,0)--(8,0)
            node[midway, above, yshift=2pt, blue, scale=0.8] {$e=(u,v)$};    
    \end{tikzpicture}
    \caption{The segment endpoints $x$ and $y$ lie in different decomposed paths $Q_i$ and $Q_j$, respectively.}
    \label{fig:xy_present_in_different_decomposed_path}
\end{figure}

Since there are $\log_{1+\epsilon} n$ segments, there can be at most $2\log_{1+\epsilon} n$ segment endpoints. 
Hence, the total time required for this operation is $O(\log_{1+\epsilon} n)$.

\section{\texorpdfstring{Proof of the \Cref{R_does_not_pass_xy}}{Proof of the result R does not pass xy}}
\label{proof of 7.1}

We analyze the case when $R$ avoids the entire segment $xy$. In 
    \ifnum\myone=1
            \Cref{querynormal:8}
        \else
            Line \ref{querynormal:8}
        \fi of $\quDtwo(t,\ro(\ft(st)))$, we move to the segment child $\xx$, which corresponds to paths that avoid $xy$. Since $R$ avoids $xy$, it lies entirely in $G-xy$, and hence survives in the associated graph $G_\xx$. Note that $G_\xx$ contains at most one faulty edge, namely the second faulty edge $e'$.

At node $\xx$, we  consider the path $\pr_\xx = P(xy)$ path for which all netpoints and segments are precomputed (see 
        \ifnum\myone=1
            \Cref{line:line7}
        \else
            Line \ref{line:line7}
        \fi in \Cref{alg:constructftgeneral}). Let us assume that the edge $e'$ lies in some segment $x'y'$ of $G_\xx$. The path $R$ can interact with $x'y'$ in three ways analogous to the single-fault setting.:
\begin{itemize}
    \item $R$ avoids $x'y'$ entirely

    In this case,  a segment child of $\xx$ contains $R$, yielding an exact ($1$-approximate) path.
    \item $R$ intersects $x'y'$ above $e'$

    The argument follows \Cref{item:R intersect above e}, giving a detour of length $(\Ep)$ and hence a $(1+k^5\ep)$-approximation of $|R|$.
    
    \item $R$ intersects $x'y'$ below $e'$

    The argument is similar to \Cref{item:R_intersect_below_e}, yielding a $(1+k^5\ep)$-approximation of $|R|$. The approximation is weaker because, unlike in \Cref{lem:k}, the induction in \Cref{lem:kk} uses a larger approximation ratio.
\end{itemize}

\end{document}